\documentclass[11pt,a4paper]{article}
\usepackage{jheppub}

\makeatletter
\def\@fpheader{\relax}
\makeatother

\allowdisplaybreaks

\newtheorem{lemma}{Lemma}[section]

\title{A perturbative Liouville prescription for the celestial three-gluon amplitude}

\author[a,b]{Grzegorz Biskowski,}
\author[b]{Franco Ferrari,}
\author[b]{Marcin R.~Pi\c{a}tek,}
\author[b]{Artur R.~Pietrykowski}
\affiliation[a]{Doctoral School University of Szczecin, Mickiewicza 18, 70--384 Szczecin, Poland}
\affiliation[b]{Institute of Physics, University of Szczecin, Wielkopolska 15, 70--451 Szczecin, Poland}

\emailAdd{biskowski.grzegorz@gmail.com}
\emailAdd{franco.ferrari@usz.edu.pl}
\emailAdd{marcin.piatek@usz.edu.pl}
\emailAdd{hearthie@gmail.com}

\abstract{
We study the celestial three-point Mellin representation for three-gluon
scattering in a dilaton background within the Mellin--Liouville
formulation of Stieberger, Taylor and Zhu (STZ). We formulate an
alternative perturbative prescription in which the light Liouville
parameters are kept fixed in the small-$b$ expansion. Requiring global
conformal covariance, a controlled light-operator expansion of the
Liouville DOZZ structure constant, and consistency with the standard
celestial three-gluon representation fixes the corresponding
Mellin--Liouville map and operator normalization. Starting from the
full Liouville three-point function, we derive the leading and first
subleading terms in the $b^2$ expansion. The leading term reproduces
the standard tree-level celestial Yang--Mills representation in the
appropriate small-total-momentum limit, while the first subleading
term is interpreted as the first coefficient of a controlled Liouville
deformation of this tree-level construction. We develop an operator
representation of the inverse Mellin transform, whereby polynomial
dependence on the Mellin parameters is converted into differential
operators acting on the leading Mellin integral. The resulting
subleading contribution is obtained in closed form in terms of
modified Bessel functions, with characteristic logarithmic and
nonanalytic behaviour at small total momentum. Finally, we clarify the
relation between the fixed-$\sigma$ and STZ finite-$b$ parametrizations
by deriving the induced nonlinear reparametrization of the Mellin
dimension and its action on the inverse Mellin representation. The
result provides a systematic framework for studying finite-$b$
Liouville deformations and their possible relation to perturbative
structures in celestial Yang--Mills theory.
}

\begin{document}
\maketitle

\section{Introduction}
Celestial holography relates four-dimensional scattering amplitudes to
correlation functions of a two-dimensional conformal field theory
living on the celestial sphere.\footnote{See, e.g.,
refs.~\cite{Strominger2017,Pasterski:2021rjz,Raclariu:2021zjz,Donnay:2023celestial,
Pasterski:2021dqe,Pasterski:2017kqt,Schreiber:2017jsr,Kalyanapuram:2020aya,
Banerjee:2020kaa,Melton:2024akx,Giribet:2024vnk,Donnay:2025yoy}
for reviews, pedagogical introductions, and seminal works.}
In this framework, four-dimensional asymptotic states are mapped to
celestial conformal operators by Mellin transforms, expressing
scattering amplitudes in terms of correlators with definite
transformation properties under the global conformal group of the
celestial sphere.

A relation between Liouville theory and celestial Yang--Mills
amplitudes in a dilaton background was originally proposed in
\cite{Stieberger2023a}. In the subsequent work
\cite{Stieberger2023b}, Stieberger, Taylor and Zhu (STZ) developed this
proposal for the three-gluon sector by expressing the corresponding
three-point amplitude as an inverse Mellin transform of a celestial
correlator of gluon operators constructed from chiral currents, light
Liouville operators, and helicity-dependent normalization factors.
Its leading term reproduces the tree-level three-gluon Yang--Mills
amplitude in the appropriate small-total-momentum limit.

The finite-$b$ extension of this construction depends on how the
Liouville parameters are related to the Mellin dimensions. In the STZ
prescription, the Liouville parameters are chosen so that the standard
celestial conformal weights are maintained at finite $b$. In the
present work, we consider a different organization of the
perturbative expansion, in which the light Liouville parameters
$\sigma_i$ are kept fixed as $b$ is varied. The corresponding Liouville
conformal weights then retain their exact $b$-dependence, while the
Mellin parameters $\Delta_i$ remain the variables of the inverse
Mellin transform.

At finite $b$, this prescription is understood as a controlled
Liouville deformation of the standard celestial three-gluon
construction, rather than as a standard finite-$b$ celestial CFT
correlator with the usual $\Delta_i$-dependent conformal weights. The
prescription is fixed by three consistency conditions. First, the
dressed Liouville three-point function must be globally conformally
covariant, with the position-dependent factors determined by the exact
Liouville conformal weights. Second, the light-operator expansion of
the DOZZ three-point function must admit a controlled small-$b$
expansion at fixed $\sigma_i$. Third, the leading term of the inverse
Mellin transform must reproduce the standard tree-level three-gluon
celestial representation. These conditions determine the
Mellin--Liouville map and the normalization used in the perturbative
prescription.

Starting from the resulting three-point correlator, we derive the
small-$b$ expansion of the Liouville deformation. At $b=0$, the
construction reduces to the standard celestial three-gluon
representation and, in the appropriate small-total-momentum limit,
to the tree-level Yang--Mills amplitude. At order $b^2$, the corrections 
to the Liouville conformal weights, together with the contributions induced 
by the small-$b$ asymptotics of the DOZZ structure constant, produce 
polynomial insertions in the Mellin variables $\Delta_i$.

These polynomial insertions can be converted into differential
operators acting on the leading Mellin integral. This gives a direct
and systematic evaluation of the first subleading inverse Mellin
transform, which can be expressed in terms of the modified Bessel
functions $K_1$ and $K_0$. We also determine its small-total-momentum
behaviour, including the logarithmic and nonanalytic terms generated
by the Bessel-function expansion.

We further clarify the relation between the fixed-$\sigma$ and STZ
finite-$b$ parametrizations. At a common value of the Liouville
parameter, equating the corresponding Liouville weights gives the
reparametrization
\begin{equation}
\widehat{\Delta}
=
\Delta+b^2 f_\pm(\Delta),
\label{eq:intro_Delta_map}
\end{equation}
where the functions $f_\pm$ are determined by the helicity-dependent
relations between $\sigma_\pm$ and $\Delta$. This relation induces
corresponding transformations of the Mellin measure, the Mellin
kernel, and the argument of the Mellin-space function. The resulting
contributions can be combined into a differential operator acting on
the Mellin-space integrand, providing an explicit operatorial relation
between the two finite-$b$ parametrizations.

The subleading coefficients obtained in the present fixed-$\sigma$
prescription are not identified here with Yang--Mills loop amplitudes.
Although the relation $b^2\propto\beta_0 g^2$ provides a possible link
between the two perturbative organizations, establishing a relation to
genuine loop amplitudes requires a separate comparison of the full
Mellin-space expressions in the two parametrizations, together with
the corresponding loop-level scattering amplitudes. We leave this
analysis for future work.

The paper is organized as follows. In section~\ref{ML_amplitude}, we
formulate the Mellin--Liouville three-point representation and
establish the consistency conditions that determine the prescription.
In section~\ref{Leading_order}, we analyze the leading term and recover
the standard tree-level Yang--Mills amplitude in the appropriate
small-total-momentum limit. Section~\ref{Subleading} is devoted to
the first subleading Liouville deformation, its operator
representation, its exact Bessel-function form, and its
small-total-momentum behaviour. In section~\ref{sec:STZ_loop}, we
compare the fixed-$\sigma$ and STZ finite-$b$ parametrizations, derive
the induced Mellin reparametrization, and formulate its action on the
inverse Mellin integral. We conclude in
section~\ref{sec:concluding_remarks}. Technical details are collected
in appendices~\ref{app:operator_proof} and
\ref{app:soft_limit_C1}.

\section{Mellin--Liouville formulation of the amplitude}
\label{ML_amplitude}
\subsection{The Mellin--Liouville map}
We start from the STZ celestial representation (inverse Mellin transform)
\cite{Stieberger2023b}:
\begin{align}\label{3GIcel}
{\cal A}_{\rm 3pt}(\omega_i,z_i,\bar z_i)
&=\left(\frac{1}{2\pi i}\right)^3\int_{c-i\infty}^{c+i\infty}
\prod_{j=1}^3 {\rm d}\Delta_j\;
M^{\sum_{j}\Delta_j-3}\;
\prod_{j=1}^3\omega_j^{-\Delta_j}\nonumber\\
&\qquad\times
\left\langle
O^{-a_1}_{\Delta_1}(z_1,\bar z_1)
O^{-a_2}_{\Delta_2}(z_2,\bar z_2)
O^{+a_3}_{\Delta_3}(z_3,\bar z_3)
\right\rangle.
\end{align}
Here $M$ denotes a renormalization scale introduced so that the
three-gluon amplitude has the correct mass dimension $-3$. All inverse
Mellin integrals are performed along vertical contours
$\Re\Delta_i=c>0$; after evaluating the integrals one takes the limit
$c\to0^+$. The quantities $\omega_i$, $z_i,\bar z_i$ and $\Delta_i$
denote, respectively, the energy, the celestial coordinates, and the
Mellin parameters associated with the $i$-th particle.

The three-point correlator entering \eqref{3GIcel} is constructed from
chiral WZW currents, light Liouville vertex operators
$\mathsf{V}_{\alpha}(z,\bar z)=\mathrm{e}^{2\alpha\phi(z,\bar z)}$,
$\alpha=b\sigma$,
and helicity-dependent normalization factors. In the construction of
ref.~\cite{Stieberger2023b}, the corresponding gluon operators are
schematically of the form
\begin{equation}
O_{\Delta}^{a\pm}(z,\bar z)
\propto
J_{\pm}^{a}\,
\mathrm{e}^{2i\sigma(\Delta\mp1)b\phi(z,\bar z)}.
\end{equation}
The parameter $\sigma$ is chosen as a function of $\Delta$ and $b$ so
that the conformal weights of the gluon operators retain their standard
celestial values at finite $b$. In this organization, the celestial
weights are kept fixed while the corresponding Liouville momenta vary
with $b$. The resulting three-point correlator factorizes into the
three-point function of the chiral currents, helicity-dependent
normalization factors, and the Liouville three-point function of light
primary operators.

In the present work, we consider an alternative perturbative
prescription in which the light Liouville parameters $\sigma_i$ are
kept fixed in the small-$b$ expansion. The finite-$b$ deformation is
constrained by three consistency conditions:
\begin{itemize}
\item[1.] the resulting dressed Liouville three-point function is
globally conformally covariant under ${\rm SL}(2,\mathbb C)$ on the celestial
sphere, with its position-dependent factors determined by the exact
Liouville weights $h(\sigma_i)$, thereby reflecting the Lorentz
covariance of the four-dimensional scattering amplitude;\footnote{In
celestial holography, the four-dimensional Lorentz group acts on the
celestial sphere as its global conformal group ${\rm SL}(2,\mathbb C)$; see,
e.g., refs.~\cite{Strominger2017,Pasterski:2021rjz,Raclariu:2021zjz}.}

\item[2.] the Liouville DOZZ three-point function admits the small-$b$
expansion for light primaries at fixed $\sigma_i$, with subleading
coefficients that are independent of $b$;

\item[3.] the inverse Mellin transform reproduces the standard
tree-level three-gluon amplitude at leading order
$\mathcal{O}(b^0)$.
\end{itemize}

These conditions fix the Mellin--Liouville map and the overall
normalization within the perturbative prescription considered here.
They are implemented by the following operators:
\begin{align}\label{O+}
\mathbb{O}^{+a}_{\Delta}(z,\bar z)
&=
\mathrm{F}_{+}(\Delta,\mu,b)\,
J^{a}(z)\,
{\sf V}_{h(\frac{1}{2}(\Delta-1))}(z,\bar z),
\\\label{O-}
\mathbb{O}^{-a}_{\Delta}(z,\bar z)
&=
\mathrm{F}_{-}(\Delta,\mu,b)\,
\hat J^{a}(z)\,
{\sf V}_{h(\frac{1}{2}(\Delta+1))}(z,\bar z).
\end{align}

The normalization factors are chosen as\footnote{Here, $\mu$ denotes
the Liouville cosmological constant, and
$\gamma(x)\equiv\Gamma(x)/\Gamma(1-x)$.}
\begin{align}\label{F+}
\mathrm{F}_{+}(\Delta,\mu,b)
&=
\left[\pi\mu\gamma(b^2)b^{-2b^2}\right]^{\frac{1}{2}(\Delta-1)}
\Gamma(\Delta-1),
\\\label{F-}
\mathrm{F}_{-}(\Delta,\mu,b)
&=
\left[\pi\mu\gamma(b^2)b^{-2b^2}\right]^{\frac{1}{2}(\Delta+1)-\frac{1}{2}}
\Gamma(\Delta+1).
\end{align}

The chiral currents have weights
$(h,\bar h)=(1,0)$ for $J^a$ and
$(h,\bar h)=(-1,0)$ for $\hat J^a$, and their three-point correlator
reads\footnote{Here, $z_{ij}\equiv z_i-z_j$, the index $a$ labels the
adjoint representation of the Lie group, and $f^{a_1a_2a_3}$ are the
structure constants.}
\begin{equation}
\left\langle
\hat J^{a_1}(z_1)\hat J^{a_2}(z_2)J^{a_3}(z_3)
\right\rangle
=
f^{a_1a_2a_3}\,
\frac{z_{12}^3}{z_{23}z_{31}}.
\end{equation}

In \eqref{O+} and \eqref{O-}, the operators
${\sf V}_{h(\sigma)}(z,\bar z)$ denote light Liouville primary vertex
operators
\begin{equation}
{\sf V}_{h(\sigma)}(z,\bar z)
=
\mathrm{e}^{2b\sigma\phi(z,\bar z)}
\end{equation}
with identical holomorphic and antiholomorphic conformal weights.
Using $\alpha=b\sigma$ and
$h(\alpha)=\alpha({\rm Q}-\alpha)$ with
${\rm Q}=b+\tfrac{1}{b}$ gives
\[
h(b\sigma)
=
b\sigma\left(b+\frac{1}{b}-b\sigma\right)
=
\sigma+b^2\sigma(1-\sigma).
\]
Thus,
\begin{equation}\label{Liou_h}
h(\sigma)=\bar h(\sigma)
=
\sigma+b^2\sigma(1-\sigma).
\end{equation}

The notation ${\sf V}_{h(\sigma)}$ emphasizes that the conformal weight
of the light Liouville vertex operator is $h(\sigma)$, while its
Liouville momentum is parametrized by $\alpha=b\sigma$.
The operators \eqref{O+} and \eqref{O-} therefore have conformal weights
\begin{align}\label{H+}
(H_+,\bar H_+)
&=
\bigl(1+h(\sigma_+),h(\sigma_+)\bigr),
\\
(H_-,\bar H_-)
&=
\bigl(-1+h(\sigma_-),h(\sigma_-)\bigr),\label{H-}
\end{align}
where
\begin{equation}
\sigma_+=\frac{\Delta-1}{2},
\qquad
\sigma_-=\frac{\Delta+1}{2}.
\end{equation}
In the semiclassical limit $b\to0$, these reduce to the standard
celestial weights
\begin{equation}
(H_\pm,\bar H_\pm)
\xrightarrow[b\to0]{}
\left(
\frac{\Delta\pm1}{2},
\frac{\Delta\mp1}{2}
\right).
\end{equation}

It is important to distinguish the Mellin parameter $\Delta$ from
the conformal weights of the dressed operators \eqref{O+}--\eqref{O-}.
The parameter $\Delta$ labels the inverse Mellin transform and is kept
independent of $b$ in the perturbative prescription adopted here.
By contrast, the Liouville conformal weight $h(\sigma)$ depends on
$b$, since the light Liouville parameter $\sigma$ is held fixed.
Consequently, within the fixed-${\sigma}$ prescription,
$\Delta$ remains the integration variable of the inverse Mellin
transform and the kernel is kept in the standard form
$\omega^{-\Delta}$. The finite-$b$ dependence is instead carried by
the Liouville correlator and by the normalization factors entering the
operators \eqref{O+}--\eqref{O-}.

The relations
\begin{equation}\label{sigDel}
\sigma_1=\frac{\Delta_1+1}{2},
\qquad
\sigma_2=\frac{\Delta_2+1}{2},
\qquad
\sigma_3=\frac{\Delta_3-1}{2}
\end{equation}
are therefore understood to be $b$-independent throughout the
perturbative expansion.

At finite $b$, this prescription should be viewed as a controlled
\textit{quantum} Liouville deformation of the standard celestial
three-gluon construction. The Mellin parameters $\Delta_i$ retain
their kinematical role in the inverse Mellin transform, while the
Liouville dressing induces the $b$-dependent corrections to the
conformal weights through $h(\sigma_i)$. The conventional celestial
interpretation is recovered in the semiclassical limit $b\to0$, where
$h(\sigma_i)\to\sigma_i$ and the dressed conformal weights reduce to
the standard celestial assignments. Thus, we do not require the
finite-$b$ prescription to define a standard celestial CFT correlator.

With these choices, the three-point correlator takes the form
\begin{equation}\label{3Gcel}
\begin{aligned}
&\left\langle
\mathbb{O}^{-a_1}_{\Delta_1}(z_1,\bar z_1)
\mathbb{O}^{-a_2}_{\Delta_2}(z_2,\bar z_2)
\mathbb{O}^{+a_3}_{\Delta_3}(z_3,\bar z_3)\right\rangle
\\[2pt]
&\quad=
\left\langle
\hat J^{a_1}(z_1)
\hat J^{a_2}(z_2)
J^{a_3}(z_3)\right\rangle
\,\mathrm{F}_{1-}\,\mathrm{F}_{2-}\,\mathrm{F}_{3+}
\\[2pt]
&\qquad\times
\left\langle
{\sf V}_{h(\frac{1}{2}(\Delta_1+1))}(z_1,\bar z_1)
{\sf V}_{h(\frac{1}{2}(\Delta_2+1))}(z_2,\bar z_2)
{\sf V}_{h(\frac{1}{2}(\Delta_3-1))}(z_3,\bar z_3)
\right\rangle
\\[2pt]
&\quad=
f^{a_1a_2a_3}\frac{z_{12}^3}{z_{23}z_{31}}\,
\mathrm{F}_{1-}\,\mathrm{F}_{2-}\,\mathrm{F}_{3+}
\\[2pt]
&\qquad\times
(z_{12}\bar z_{12})^{h(\sigma_3)-h(\sigma_1)-h(\sigma_2)}
(z_{13}\bar z_{13})^{h(\sigma_2)-h(\sigma_1)-h(\sigma_3)}
\\[2pt]
&\qquad\times
(z_{23}\bar z_{23})^{h(\sigma_1)-h(\sigma_2)-h(\sigma_3)}
\,C(b\sigma_1,b\sigma_2,b\sigma_3),
\end{aligned}
\end{equation}
where $C(\alpha_1,\alpha_2,\alpha_3)$ is the DOZZ structure constant,
\begin{eqnarray}
\label{DOZZ}
C(\alpha_1, \alpha_2, \alpha_3)
&=&
\left[\pi\mu\gamma(b^2)b^{2-2b^2}\right]^{
({\rm Q}-\alpha_1-\alpha_2-\alpha_3)/b}
\times\nonumber\\
&&\hspace{-60pt}
\frac{\Upsilon_{0}\Upsilon_{b}(2\alpha_1)
\Upsilon_{b}(2\alpha_2)\Upsilon_{b}(2\alpha_3)}
{\Upsilon_{b}(\sum_{i}\alpha_i-{\rm Q})
\Upsilon_{b}(\alpha_1+\alpha_2-\alpha_3)
\Upsilon_{b}(\alpha_2+\alpha_3-\alpha_1)
\Upsilon_{b}(\alpha_3+\alpha_1-\alpha_2)}.
\end{eqnarray}
The structure constant~\eqref{DOZZ} is expressed in terms of the
special function $\Upsilon_b(x)$, which can be constructed from Barnes'
double Gamma function.\footnote{For background on two-dimensional
conformal field theory and for details of Liouville theory, see
refs.~\cite{BPZ,S,DO,T1,ZZ5,T3,T4,Nakayama:2004vk,HMW,SR14}
and references therein.}

For later use, the conformal weights appearing in \eqref{3Gcel}
can be written explicitly as
\begin{equation}\label{hij}
\begin{aligned}
h_1\equiv h(\sigma_1)
&=
\frac{\Delta_1+1}{2}
+\frac{b^2}{4}(\Delta_1+1)(1-\Delta_1),
\\
h_2\equiv h(\sigma_2)
&=
\frac{\Delta_2+1}{2}
+\frac{b^2}{4}(\Delta_2+1)(1-\Delta_2),
\\
h_3\equiv h(\sigma_3)
&=
\frac{\Delta_3-1}{2}
+\frac{b^2}{4}(\Delta_3-1)(3-\Delta_3).
\end{aligned}
\end{equation}

We finally comment on the relation between the present prescription
and that of ref.~\cite{Stieberger2023b}. In ref.~\cite{Stieberger2023b},
the Liouville parameter $\sigma$ may be chosen with a $b$-dependent
relation to the Mellin dimension $\Delta$, so that the conformal
weights of the corresponding gluon operators remain in the standard
celestial form also away from the semiclassical limit. In the present
work, by contrast, the light Liouville parameters $\sigma_i$ are kept
fixed according to eq.~(\ref{sigDel}). Their Liouville conformal
weights therefore acquire the explicit $b$-dependent corrections
displayed in eq.~\eqref{hij}, which are retained in the
position-dependent prefactor of the Liouville three-point function.

The two prescriptions agree on the semiclassical Mellin--Liouville
identification, but differ in the finite-$b$ parametrization of the
deformation. The present choice is tailored to the perturbative
expansion of the DOZZ structure constant at fixed light Liouville
parameters.

Moreover, the two prescriptions can be related more explicitly by a
nonlinear reparametrization of the Mellin dimension.\footnote{For
details, see section~\ref{sec:STZ_loop}.}

\subsection{Perturbative expansion}
Ref.~\cite{Stieberger2023b} explored the possibility of extracting
subleading terms in the three-gluon Mellin--Liouville representation by
expanding the inverse Mellin integrand in powers of $b^2$. As emphasized
there, however, a systematic treatment requires several nontrivial
ingredients, including the small-$b$ asymptotics of the
$\Upsilon_b$-function, the normalization factors
$\mathrm{F}_{\pm}$, and the relation between the Liouville parameters
$\sigma_i$ and the Mellin dimensions $\Delta_i$.

In the present work, these ingredients are specified within the
perturbative prescription introduced above. This allows us to organize
the finite-$b$ construction as a systematic deformation of the
standard celestial three-gluon representation, with the
$b^0$ contribution identified with the tree-level celestial
Yang--Mills amplitude.

In the correlator~\eqref{3Gcel}, there are three distinct sources of
$b$-dependence:
(i) the normalization factors
$\mathrm{F}_{1-}\mathrm{F}_{2-}\mathrm{F}_{3+}$;
(ii) the position-dependent factor
$\prod_{i<j}|z_{ij}|^{2(h_k-h_i-h_j)}$;
and (iii) the DOZZ structure constant
$C\!\left(b\sigma_1,b\sigma_2,b\sigma_3\right)$.

The normalization factors are kept unexpanded. Their exact
$b$-dependence is retained in the overall prefactor defined below,
while the perturbative expansion is performed on the remaining
Liouville-dependent part of the correlator. With this organization,
the leading term and its subleading Liouville deformations can be
treated on the same footing.

We first expand the factor depending on $z_{ij}$. Using
eqs.~\eqref{hij}, we find
\begin{eqnarray}
h_3-h_1-h_2&=& \frac{1}{2}\left(\Delta_3-\Delta_1-\Delta_2-3\right)
+\frac{b^2}{4}\left(\Delta_{1}^{2}+\Delta_{2}^{2}-\Delta_{3}^{2}+4\Delta_{3}-5\right),
\nonumber\\
h_2-h_1-h_3&=& \frac{1}{2}\left(\Delta_2-\Delta_1-\Delta_3+1\right)
+\frac{b^2}{4}\left(\Delta_{1}^{2}+\Delta_{3}^{2}-\Delta_{2}^{2}-4\Delta_{3}+3\right),
\nonumber\\
h_1-h_2-h_3&=& \frac{1}{2}\left(\Delta_1-\Delta_2-\Delta_3+1\right)
+\frac{b^2}{4}\left(\Delta_{3}^{2}+\Delta_{2}^{2}-\Delta_{1}^{2}-4\Delta_{3}+3\right).
\end{eqnarray}

Consequently, the position-dependent factor admits the expansion
\begin{equation}\label{Zs}
\begin{aligned}
&(z_{12}\bar z_{12})^{h_3-h_1-h_2}
(z_{13}\bar z_{13})^{h_2-h_1-h_3}
(z_{23}\bar z_{23})^{h_1-h_2-h_3}
\\
&\hspace{40pt}
= a_0+a_1b^2+a_2b^4+a_3b^6+\cdots,
\end{aligned}
\end{equation}
where
\begin{eqnarray}\label{a0}
a_0(\Delta_i;z_i,\bar z_i) &=&
(z_{12}\bar z_{12})^{\tfrac{1}{2}(\Delta_3-\Delta_1-\Delta_2-3)}
(z_{13}\bar z_{13})^{\tfrac{1}{2}(\Delta_2-\Delta_1-\Delta_3+1)}
\nonumber\\
&\times&
(z_{23}\bar z_{23})^{\tfrac{1}{2}(\Delta_1-\Delta_2-\Delta_3+1)},
\\[6pt]
T(\Delta_i;z_i,\bar z_i)&=&
(\Delta_{1}^{2}+\Delta_{2}^{2}-\Delta_{3}^{2}+4\Delta_{3}-5)\,
\ln(z_{12}\bar z_{12})
\nonumber\\
&+&
(\Delta_{1}^{2}+\Delta_{3}^{2}-\Delta_{2}^{2}-4\Delta_{3}+3)\,
\ln(z_{13}\bar z_{13})
\nonumber\\
&+&
(\Delta_{3}^{2}+\Delta_{2}^{2}-\Delta_{1}^{2}-4\Delta_{3}+3)\,
\ln(z_{23}\bar z_{23}) .
\end{eqnarray}
Thus
\begin{equation}\label{ana0T}
a_1=\tfrac{1}{4}\,a_0\,T,\qquad
a_2=\tfrac{1}{32}\,a_0\,T^2,\qquad
a_3=\tfrac{1}{384}\,a_0\,T^3.
\end{equation}

We next consider the DOZZ structure constant. Using the small-$b$
asymptotic expansion of the $\Upsilon_b$-function~\cite{Thorn},
\begin{equation}\label{Ub}
\Upsilon_b(b\sigma)
\sim
b\,\Upsilon_{0}\,
\frac{b^{\,b^2\sigma(1-\sigma)-\sigma}}{\Gamma(\sigma)}
\exp\left[
b^2\gamma_{E}\,\phi_{1}(\sigma)
+\sum_{n=1}^{\infty} b^{2(2n+1)}
\frac{\zeta(2n+1)}{2n+1}\,\phi_{2n+1}(\sigma)
\right],
\end{equation}
one obtains the corresponding asymptotic form\footnote{Here,
\(\pi\tilde\mu\gamma(b^{-2})=\big(\pi\mu\gamma(b^{2})\big)^{b^{-2}}\).}
\begin{equation}\label{Cim}
C\!\left(b\sigma_1,b\sigma_2,b\sigma_3\right)
\sim
P\!\left(b;\sigma_1,\sigma_2,\sigma_3\right)
A(\sigma_1,\sigma_2,\sigma_3)
\left(1+\sum_{n\ge1}b^{2n}\,\Omega_n(\sigma_1,\sigma_2,\sigma_3)\right),
\end{equation}
where
\begin{equation}
P\!\left(b;\sigma_1,\sigma_2,\sigma_3\right)=
\frac{\pi\tilde\mu\,\gamma(b^{-2})
\left(\pi\mu\,\gamma(b^2)b^{-2b^2}\right)^{1-\sum_i\sigma_i}
\gamma\left(\sum_i\sigma_i-1-\frac{1}{b^2}\right)
b^{-2b^2(\sum_i\sigma_i-1)}}{b^{5}}
\end{equation}
and
\begin{equation}
A(\sigma_1,\sigma_2,\sigma_3)=
\frac{\Gamma(\sigma_{1}+\sigma_{2}+\sigma_{3}-1)
\Gamma(\sigma_{2}+\sigma_{3}-\sigma_{1})
\Gamma(\sigma_{3}+\sigma_{1}-\sigma_{2})
\Gamma(\sigma_{1}+\sigma_{2}-\sigma_{3})}
{\Gamma(2\sigma_{1})\Gamma(2\sigma_{2})\Gamma(2\sigma_{3})}.
\end{equation}

The coefficients $\Omega_n$ can be computed analytically in a fully
controlled manner.\footnote{See~\cite{Ferrari2025}.}
They are symmetric polynomials in $\sigma_1,\sigma_2,\sigma_3$.
For example,
\begin{eqnarray}
\Omega_1(\sigma_1,\sigma_2,\sigma_3)
&=&-2\gamma_{E}(\sigma_1+\sigma_2+\sigma_3-1),
\\[5pt]
\Omega_2(\sigma_1,\sigma_2,\sigma_3)
&=&2\gamma^{2}_{E}(\sigma_1+\sigma_2+\sigma_3-1)^2,
\end{eqnarray}
while
\begin{eqnarray}
\Omega_3(\sigma_1,\sigma_2,\sigma_3)
&=& \frac{1}{6} \Bigg(
-8\gamma^{3}_{E}(\sigma_1+\sigma_2+\sigma_3-1)^3
\nonumber\\
&&\quad
-4(\sigma_1+\sigma_2+\sigma_3-1)\Big(
1 + 3\sigma_1^3 - 2\sigma_2 - 2\sigma_3
\nonumber\\
&&\quad
+\sigma_1^2(1 - 3\sigma_2 - 3\sigma_3)
+\sigma_1\big(-2+(2 - 3\sigma_2)\sigma_2
\nonumber\\
&&\qquad
+(2 + 6\sigma_2)\sigma_3 - 3\sigma_3^2\big)
\nonumber\\
&&\quad
+(\sigma_2 + \sigma_3)
\left(\sigma_2 + 3\sigma_2^2 + \sigma_3
-6\sigma_2\sigma_3 + 3\sigma_3^2\right)
\Big)\zeta(3)
\Bigg).
\end{eqnarray}
The expansion coefficients involve the Euler--Mascheroni constant
$\gamma_E$ and values of the Riemann zeta function.

The function $A(\sigma_1,\sigma_2,\sigma_3)$ is the leading light
semiclassical contribution to the DOZZ structure constant. The
higher-order terms in the asymptotic expansion are encoded in the
coefficients $\Omega_n$.

For later convenience, we factor out from
$P(b;\sigma_1,\sigma_2,\sigma_3)$ the combination
$\pi\tilde\mu/b$ and absorb it into the leading light factor.
This gives
\begin{equation}
\mathcal{N}\!\left(b;\sigma_1,\sigma_2,\sigma_3\right)
=
\frac{
\gamma(b^{-2})
\left(\pi\mu\,\gamma(b^2)b^{-2b^2}\right)^{1-\sum_i\sigma_i}
\gamma\left(\sum_i\sigma_i-1-\frac{1}{b^2}\right)
b^{-2b^2(\sum_i\sigma_i-1)}
}{b^{4}}.
\end{equation}
With this convention, the factor multiplying the perturbative
expansion below is
\begin{equation}
\mathrm{PF}\equiv
f^{a_1a_2a_3}\,
\frac{z_{12}^3}{z_{23}z_{31}}\,
\mathrm{F}_{1-}\,\mathrm{F}_{2-}\,\mathrm{F}_{3+}\,
\frac{\pi\tilde\mu}{b}\,A .
\end{equation}

Substituting \eqref{Zs} and \eqref{Cim} into
\eqref{3Gcel}, the dressed Liouville three-point function can then be
organized asymptotically as
\begin{equation}
\left\langle 
\mathbb{O}^{-a_1}_{\Delta_1}
\mathbb{O}^{-a_2}_{\Delta_2}
\mathbb{O}^{+a_3}_{\Delta_3}\right\rangle
\sim
\mathcal{N}\,
\big(
\tilde{\mathcal A}^{(0)}
+b^2\tilde{\mathcal A}^{(1)}
+b^4\tilde{\mathcal A}^{(2)}
+\cdots
\big),
\end{equation}
where
\begin{equation}\label{AAA}
\tilde{\mathcal A}^{(0)}
=
a_0\,\mathrm{PF},
\qquad
\tilde{\mathcal A}^{(1)}
=
a_0\left(\tfrac{1}{4}T+\Omega_1\right)\mathrm{PF},
\qquad
\tilde{\mathcal A}^{(2)}
=
a_0\left(
\tfrac{1}{32}T^2
+\tfrac{1}{4}T\Omega_1
+\Omega_2
\right)\mathrm{PF}.
\end{equation}

The coefficients in \eqref{AAA} should be viewed as successive
coefficients of the Liouville deformation. In particular, the zeroth
order term gives the standard celestial three-gluon representation,
whereas the higher-order terms encode the corrections generated by
the finite-$b$ Liouville dressing.

The inverse Mellin transform defined in \eqref{3GIcel}, applied to the
leading contribution 
$\bigl(\pi\mu\gamma(b^2)b^{-2b^2}\bigr)^{1-\sum_i\sigma_i}\tilde{\mathcal A}^{(0)}$,
reproduces the STZ tree-level amplitude
\cite{Stieberger2023b}:
\begin{equation}\label{A0_repeat}\boxed{\;
{\cal A}_{\rm 3pt}^{\rm tree}
\equiv
{\cal A}_{\rm 3pt}^{(0)}(\omega_i,z_i,\bar z_i)
=
\frac{\pi}{b}\frac{\tilde\mu}{M^2}
f^{a_1a_2a_3}
\frac{z_{12}^3}{z_{23}\,z_{31}}
{\cal I}^{(0)}(\omega_{1},\omega_{2},\omega_{3}),}
\end{equation}
with
\begin{eqnarray}
\label{I0_repeat}
{\cal I}^{(0)}(\omega_{1},\omega_{2},\omega_{3})
&=&
\left(\frac{1}{2\pi i}\right)^3
\int_{c-i\infty}^{c+i\infty}
{\rm d}\Delta_1{\rm d}\Delta_2{\rm d}\Delta_3\;
M^{\Delta_1+\Delta_2+\Delta_3-1}
\prod_{i=1}^3\omega_{i}^{-\Delta_i}
\\[4pt]
&&\times
\Gamma\Big(
\frac{\Delta_1+\Delta_2+\Delta_3-1}{2}
\Big)
\Gamma\Big(
\frac{\Delta_1+\Delta_3-\Delta_2-1}{2}
\Big)
\nonumber\\[4pt]
&&\times
\Gamma\Big(
\frac{\Delta_2+\Delta_3-\Delta_1-1}{2}
\Big)
\Gamma\Big(
\frac{\Delta_1+\Delta_2-\Delta_3+3}{2}
\Big)\,
a_0(\Delta_i;z_i,\bar z_i).\nonumber
\end{eqnarray}
Here $a_0$ is the leading kinematic prefactor given in
eq.~\eqref{a0}. This establishes the identification of the
$b^0$ term with the standard tree-level celestial Yang--Mills
amplitude.

At the next orders, the Mellin integrand is modified by the
corresponding coefficients of the Liouville deformation. For example,
the first correction is obtained by replacing
\begin{equation}
a_0
\;\;\longrightarrow\;\;
a_0\left(\tfrac{1}{4}T+\Omega_1\right),
\end{equation}
while at order $b^4$ the insertion is
\begin{equation}
a_0
\left(
\tfrac{1}{32}T^2
+\tfrac{1}{4}T\Omega_1
+\Omega_2
\right).
\end{equation}
More generally, the coefficient at order $b^{2n}$ is obtained by
multiplying the expansion of the position-dependent factor with the
DOZZ expansion,
\begin{equation}
\left(a_0+a_1b^2+a_2b^4+a_3b^6+\cdots\right)
\left(1+b^2\Omega_1+b^4\Omega_2+b^6\Omega_3+\cdots\right).
\end{equation}
Thus the complete integrand at each order is determined
systematically by the light Liouville data.

The structure of these corrections admits a useful interpretation at
the level of the Mellin transform. Since the quantities
$\Omega_n(\sigma_i)$ are polynomials in the Mellin parameters
$\Delta_i$, and the coefficients generated by the expansion of the
conformal prefactor are likewise polynomial in $\Delta_i$, their
inverse Mellin transforms can be represented by differential operators
acting on the leading celestial amplitude. In particular,
\begin{equation}
\Delta_i\,\omega_i^{-\Delta_i}
=
-\frac{\partial}{\partial\log\omega_i}
\omega_i^{-\Delta_i},
\end{equation}
so that polynomial insertions in the $\Delta_i$ become differential
operators in the logarithms of the external energies. The logarithmic
dependence on $z_{ij}\bar z_{ij}$ arising from the dressed conformal
weights is retained explicitly in these deformation terms.

Accordingly, the perturbative expansion takes the schematic form
\begin{equation}
{\cal A}_{\rm 3pt}^{\rm L}(b)
=
{\cal A}_{\rm 3pt}^{\rm tree}
+
\sum_{n\geq1} b^{2n}\,
\delta{\cal A}_{\rm 3pt,L}^{(n)},
\end{equation}
where $\delta{\cal A}^{(n)}_{\rm L}$ denotes the $n$-th coefficient of
the Liouville deformation. 
These terms are not identified here with loop corrections to the
Yang--Mills amplitude. Rather, they provide a systematic tower of
Liouville-induced deformations of the tree-level celestial
representation, determined by the small-$b$ asymptotics of the DOZZ
structure constant and the corresponding $b$-dependent conformal
prefactors. The relation between this organization of the perturbative
series and the STZ parametrization is discussed further in section~\ref{sec:STZ_loop}.

\section{Leading-order structure}
\label{Leading_order}
\subsection{Inverse Mellin representation at leading order}
We calculate the inverse Mellin integral
\begin{eqnarray}
\label{I02}
{\cal I}^{(0)}(\omega_{1},\omega_{2},\omega_{3})
&=&\left(\frac{1}{2\pi i}\right)^3
\int_{c-i\infty}^{c+i\infty}\!{\rm d}\Delta_1{\rm d}\Delta_2{\rm d}\Delta_3\;
M^{\Delta_1+\Delta_2+\Delta_3-1}\,
\omega_1^{-\Delta_1}\omega_2^{-\Delta_2}\omega_3^{-\Delta_3}
\nonumber\\[4pt]
&\times&
\Gamma\Big(\frac{\Delta_1+\Delta_2+\Delta_3-1}{2}\Big)\,
\Gamma\Big(\frac{\Delta_1+\Delta_3-\Delta_2-1}{2}\Big)\,
\\[4pt]
&\times&
\Gamma\Big(\frac{\Delta_2+\Delta_3-\Delta_1-1}{2}\Big)\,
\Gamma\Big(\frac{\Delta_1+\Delta_2-\Delta_3+3}{2}\Big)
\nonumber\\[4pt]
&\times&
(z_{12}\bar z_{12})^{\frac{1}{2}(\Delta_3-\Delta_1-\Delta_2-3)}
(z_{23}\bar z_{23})^{\frac{1}{2}(\Delta_1-\Delta_2-\Delta_3+1)}
(z_{13}\bar z_{13})^{\frac{1}{2}(\Delta_2-\Delta_1-\Delta_3+1)},\nonumber
\end{eqnarray}
which gives the zeroth-order contribution to the three-point Mellin
representation, cf.~eq.~\eqref{I0_repeat}. The triple integral
\eqref{I02} has the same functional form as the integral studied in
\cite{Stieberger2023b}. In the present construction, however, it
arises directly as the leading term of the systematic small-$b$
expansion of the dressed Liouville three-point correlator
\eqref{3Gcel}. It therefore provides the starting point for the
perturbative analysis of the Liouville deformation developed below.

The integral representation of the Gamma function,
\begin{equation}\label{Ga}
\Gamma(z)=\int\limits_{0}^{\infty}{\rm d}t\,{\rm e}^{-t}t^{z-1},
\end{equation}
proves useful for rewriting the integral ${\cal I}^{(0)}$ in the form
\begin{eqnarray}
{\cal I}^{(0)}(\omega_{1},\omega_{2},\omega_{3})
&=&
\frac{1}{M}\left(\frac{1}{2\pi i}\right)^3
\int\limits_{c-i\infty}^{c+i\infty}
{\rm d}\Delta_1{\rm d}\Delta_2{\rm d}\Delta_3\;
\int\limits_{0}^{\infty}{\rm d}t_0{\rm d}t_1{\rm d}t_2{\rm d}t_3
\nonumber\\[5pt]
&\times&
{\rm e}^{\Delta_1 x_1}\,{\rm e}^{\Delta_2 x_2}\,{\rm e}^{\Delta_3 x_3}
\,{\rm e}^{-t_0-t_1-t_2-t_3}\,
\frac{t_{0}^{-\frac{1}{2}}t_{1}^{-\frac{1}{2}}
t_{2}^{-\frac{1}{2}}t_{3}^{\frac{3}{2}}}
{t_0t_1t_2t_3}
\nonumber\\[5pt]
&\times&
(z_{12}\bar z_{12})^{-\frac{3}{2}}
(z_{23}\bar z_{23})^{\frac{1}{2}}
(z_{13}\bar z_{13})^{\frac{1}{2}},
\end{eqnarray}
where
\begin{eqnarray}
x_1 &=&
\frac{1}{2}\ln\left(
\frac{M^2t_0t_1t_3z_{23}\bar z_{23}}
{\omega_{1}^{2}t_2z_{12}\bar z_{12}z_{13}\bar z_{13}}
\right),
\nonumber\\[2pt]
x_2 &=&
\frac{1}{2}\ln\left(
\frac{M^2t_0t_2t_3z_{13}\bar z_{13}}
{\omega_{2}^{2}t_1z_{12}\bar z_{12}z_{23}\bar z_{23}}
\right),
\nonumber\\[2pt]
x_3 &=&
\frac{1}{2}\ln\left(
\frac{M^2t_0t_1t_2z_{12}\bar z_{12}}
{\omega_{3}^{2}t_3z_{23}\bar z_{23}z_{13}\bar z_{13}}
\right)
\end{eqnarray}
and
\begin{equation}\label{ttt}
t_1=\frac{\omega_1\omega_3{\rm e}^{x_1+x_3}|z_{13}|^2}{M^2t_0},
\quad
t_2=\frac{\omega_2\omega_3{\rm e}^{x_2+x_3}|z_{23}|^2}{M^2t_0},
\quad
t_3=\frac{\omega_1\omega_2{\rm e}^{x_1+x_2}|z_{12}|^2}{M^2t_0}.
\end{equation}

Let us observe that the integrations over $\Delta_1$, $\Delta_2$,
and $\Delta_3$ --- that is, the inverse Mellin transforms --- yield
Dirac delta distributions. Indeed, consider a representative integral
of the form
\begin{equation}
\frac{1}{2\pi i}
\int\limits_{c-i\infty}^{c+i\infty}
{\rm e}^{\Delta x}\,{\rm d}\Delta,
\end{equation}
where $\Delta=c+ip$ with $c>0$ and $p\in\mathbb{R}$. Substituting
${\rm d}\Delta=i\,{\rm d}p$, we obtain
\begin{align}
\frac{1}{2\pi i}
\int\limits_{c-i\infty}^{c+i\infty}
{\rm e}^{\Delta x}\,{\rm d}\Delta
&=
\frac{1}{2\pi}
\int\limits_{-\infty}^{\infty}
{\rm e}^{(c+ip)x}\,{\rm d}p
\nonumber\\
&=
{\rm e}^{cx}
\frac{1}{2\pi}
\int\limits_{-\infty}^{\infty}
{\rm e}^{ipx}\,{\rm d}p
\nonumber\\
&=
{\rm e}^{cx}\delta(x),
\end{align}
where we used the standard Fourier representation of the Dirac delta
function. Thus, each inverse Mellin transform localizes the
integrand at the corresponding saddle condition $x_i=0$, while the
factor ${\rm e}^{cx_i}$ evaluates to unity on the support of the
delta function.

Simultaneously, one can perform a change of variables
$t_1,t_2,t_3\mapsto x_1,x_2,x_3$, which leads to
\begin{eqnarray}\label{I0m}
{\cal I}^{(0)}(\omega_{1},\omega_{2},\omega_{3})
&=&
\frac{2}{M}
\int\limits_{0}^{\infty}{\rm e}^{-t_0}{\rm d}t_0
\int\limits_{-\infty}^{\infty}{\rm d}x_1{\rm d}x_2{\rm d}x_3\,
{\rm e}^{cx_1}{\rm e}^{cx_2}{\rm e}^{cx_3}
\delta(x_1)\delta(x_2)\delta(x_3)
\nonumber\\[5pt]
&&\;\hspace{-40pt}\times\;
\exp\left(
-\frac{
\omega_1\omega_3{\rm e}^{x_1+x_3}|z_{13}|^2+
\omega_2\omega_3{\rm e}^{x_2+x_3}|z_{23}|^2+
\omega_1\omega_2{\rm e}^{x_1+x_2}|z_{12}|^2
}{M^2 t_0}
\right)
\nonumber\\[5pt]
&&\;\hspace{-40pt}\times\;
\frac{\omega_1\omega_2}{\omega_3 M}\,
t_{0}^{-2}\,{\rm e}^{x_1+x_2-x_3}
\nonumber\\[5pt]
&=&
\frac{2\omega_1\omega_2}{\omega_3 M^2}
\int\limits_{0}^{\infty}
{\rm e}^{-t_0}\,t_{0}^{-2}
\exp\left(
-\frac{
\omega_1\omega_3|z_{13}|^2+
\omega_2\omega_3|z_{23}|^2+
\omega_1\omega_2|z_{12}|^2
}{M^2 t_0}
\right){\rm d}t_0.
\nonumber
\end{eqnarray}

The expression in \eqref{I0m} follows from the transformation of the
measure
\begin{eqnarray}
{\rm d}t_1{\rm d}t_2{\rm d}t_3
&=&
|\det J|\,{\rm d}x_1{\rm d}x_2{\rm d}x_3
\nonumber\\
&=&
2t_1t_2t_3\,{\rm d}x_1{\rm d}x_2{\rm d}x_3,
\end{eqnarray}
where
\begin{equation}
J=
\dfrac{\partial(t_1,t_2,t_3)}
{\partial(x_1,x_2,x_3)}
=
\begin{pmatrix}
t_1&0&t_1\\
0&t_2&t_2\\
t_3&t_3&0
\end{pmatrix},
\end{equation}
together with use of \eqref{ttt} and the identity
\begin{equation}
t_{0}^{-\frac{3}{2}}t_{1}^{-\frac{1}{2}}
t_{2}^{-\frac{1}{2}}t_{3}^{\frac{3}{2}}
(z_{12}\bar z_{12})^{-\frac{3}{2}}
(z_{23}\bar z_{23})^{\frac{1}{2}}
(z_{13}\bar z_{13})^{\frac{1}{2}}
=
\dfrac{\omega_1\omega_2}{\omega_3 M}\,
t_{0}^{-2}\,
{\rm e}^{x_1+x_2-x_3}.
\end{equation}

Following~\cite{Stieberger2023b}, let us define
\begin{equation}
Q^2\equiv
\omega_1\omega_3 |z_{13}|^2+
\omega_2\omega_3 |z_{23}|^2+
\omega_1\omega_2 |z_{12}|^2.
\end{equation}
With this, the leading-order integral becomes
\begin{equation}
{\cal I}^{(0)}(\omega_{1}, \omega_{2}, \omega_{3})
=
\frac{2\omega_1\omega_2}{\omega_3 M^2}
\int\limits_0^\infty
{\rm d}t_0\,
{\rm e}^{-t_0-\frac{Q^2}{M^2t_0}}
t_0^{-2}.
\end{equation}

In spinor-helicity variables, the kinematics is encoded as
\begin{equation}
\langle ij \rangle
=
\sqrt{\omega_i\omega_j}\,z_{ij},
\qquad
[ij]
=
\sqrt{\omega_i\omega_j}\,\bar z_{ij},
\qquad
2\,p_i\cdot p_j
=
\langle ij\rangle[ij],
\qquad
p_i^2=0,
\end{equation}
from which it follows that
\[
Q^2=(p_1+p_2+p_3)^2.
\]
Thus $Q$ is the total momentum scale associated with the
three-particle configuration.

Finally, the remaining integral evaluates to a modified Bessel
function:\footnote{See \S10.32 of the NIST Digital Library of
Mathematical Functions,
\url{https://dlmf.nist.gov/10.32},
\[
K_{\nu}(z)=\tfrac{1}{2}
\Big(\tfrac{1}{2}z\Big)^{\nu}
\int_{0}^{\infty}
\exp\left(-t-\frac{z^2}{4t}\right)
\frac{{\rm d}t}{t^{\nu+1}}.
\]}
\begin{equation}
\int\limits_0^\infty
{\rm d}t_0\,
{\rm e}^{-t_0-\frac{Q^2}{M^2t_0}}
t_0^{-2}
=
\frac{2M}{Q}\,
K_1\!\left(\frac{2Q}{M}\right).
\end{equation}
The full ${\cal I}^{(0)}$ is therefore
\begin{equation}\boxed{
{\cal I}^{(0)}(\omega_{1},\omega_{2},\omega_{3})
=
\frac{4\omega_1\omega_2}{\omega_3 M Q}\,
K_1\!\left(\frac{2Q}{M}\right).
}
\label{0order}
\end{equation}

\subsection{Zeroth-order result and Liouville--Yang--Mills correspondence}
The zeroth-order contribution to the three-point Mellin
representation can therefore be written as
\begin{equation}\label{A0f}\boxed{
{\cal A}_{\rm 3pt}^{(0)}(\omega_i,z_i,\bar z_i)
=
\frac{4\pi\tilde\mu}{b M^4}
f^{a_1a_2a_3}\,
\frac{\langle 12\rangle^3}
{\langle 23\rangle\langle 31\rangle}\,
\frac{M}{Q}\,
K_1\!\left(\frac{2Q}{M}\right).
}
\end{equation}

In the soft limit $Q\to0$, the modified Bessel function in
\eqref{A0f} behaves as
\begin{equation}
\frac{2M}{Q}\,
K_1\!\left(\frac{2Q}{M}\right)
=
\frac{M^2}{Q^2}+\ldots\,,
\end{equation}
which gives
\begin{equation}\label{AYM}
{\cal A}_{\rm 3pt}^{(0)}(\omega_i,z_i,\bar z_i)
=
\frac{2\pi\tilde\mu}{bM^2Q^2}
f^{a_1a_2a_3}
\frac{\langle 12\rangle^3}
{\langle 23\rangle\langle 31\rangle}
+\ldots\,.
\end{equation}

This result can be directly compared with the standard tree-level
three-gluon contribution in Yang--Mills theory
(cf.~\cite{Stieberger2023b} and references therein):
\begin{equation}\label{AYMp}
{\cal A}_{\rm 3pt}^{\prime(0)}(\omega_i,z_i,\bar z_i)
=
\frac{g}{\Lambda\Lambda'}\,
f^{a_1a_2a_3}
\frac{1}{Q^2}
\frac{\langle 12\rangle^3}
{\langle 23\rangle\langle 31\rangle}
+\ldots\,,
\end{equation}
where $g$ is the Yang--Mills coupling constant,
$\Lambda^{-1}$ denotes the canonical coupling of the dilaton to the
gauge field strength, and $\Lambda'$ controls the strength of a
point-like dilaton source, given by
${\cal J}(x)=\delta^{(4)}(x)/\Lambda'$. Comparing
\eqref{AYM} and \eqref{AYMp} gives
\begin{equation}
g M^2(\Lambda\Lambda')^{-1}
=
2\pi\tilde\mu\,b^{-1},
\end{equation}
which is the parameter relation proposed in
\cite{Stieberger2023b}.

The result above establishes the leading-order correspondence relevant
for the present construction: the zeroth-order term in the Liouville
deformation reproduces the standard tree-level three-gluon scattering
amplitude. Recall that, on the celestial side, the semiclassical limit
$b\to0$ corresponds to
\begin{equation}
c=1+6\left(b+\frac{1}{b}\right)^2
\sim \frac{6}{b^2}\longrightarrow\infty,
\end{equation}
in which the Liouville dressing of the conformal weights disappears,
\begin{equation}
h(\sigma_\pm)\longrightarrow \sigma_\pm,
\end{equation}
and the fixed-$\sigma$ Mellin representation reduces to the conventional
celestial representation of the tree-level Yang--Mills amplitude.

\section{Subleading structure of the Liouville deformation}
\label{Subleading}

\subsection{Operator formulation of the subleading Mellin integral}
\label{inverse_mellin_subleading}
Applying the inverse Mellin transform \eqref{3GIcel} to
%\begin{equation}
$\bigl(\pi\mu\gamma(b^2)b^{-2b^2}\bigr)^{1-\sum_i\sigma_i}
\tilde{\mathcal A}^{(1)}$
%\end{equation}
defines the first subleading coefficient of the Liouville deformation,
\begin{equation}\label{A1_restate}\boxed{\;
\delta{\cal A}_{\rm 3pt,L}^{(1)}
(\omega_i,z_i,\bar z_i)
=
\frac{\pi}{b}\,\frac{\tilde\mu}{M^2}\,
f^{a_1a_2a_3}\,
\frac{z_{12}^3}{z_{23}\,z_{31}}\,
{\cal I}^{(1)}(\omega_1,\omega_2,\omega_3),\;}
\end{equation}
where
\begin{eqnarray}
\label{I1_restate}
&&{\cal I}^{(1)}(\omega_{1},\omega_{2},\omega_{3})
\;=\;
\left(\frac{1}{2\pi i}\right)^3
\int_{c-i\infty}^{c+i\infty}\!{\rm d}\Delta_1{\rm d}\Delta_2{\rm d}\Delta_3\;
M^{\Delta_1+\Delta_2+\Delta_3-1}\,\omega_1^{-\Delta_1}\omega_2^{-\Delta_2}\omega_3^{-\Delta_3}
\nonumber\\[5pt]
&&\hspace{10pt}\times\;
\Gamma\Big(\frac{\Delta_1+\Delta_2+\Delta_3-1}{2}\Big)\,
\Gamma\Big(\frac{\Delta_1+ \Delta_3-\Delta_2-1}{2}\Big)
\nonumber\\[4pt]
&&\hspace{10pt}\times\;
\Gamma\Big(\frac{\Delta_2+\Delta_3-\Delta_1-1}{2}\Big)\,
\Gamma\Big(\frac{\Delta_1+\Delta_2-\Delta_3+3}{2}\Big)
\nonumber\\[5pt]
&&\hspace{10pt}\times\;
\rho_{12}^{\frac{1}{2}(\Delta_3-\Delta_1-\Delta_2-3)}
\rho_{23}^{\frac{1}{2}(\Delta_1-\Delta_2-\Delta_3+1)}
\rho_{13}^{\frac{1}{2}(\Delta_2-\Delta_1-\Delta_3+1)}
\nonumber\\[5pt]
&&\hspace{10pt}\times\;
\Bigg\{\frac{1}{4}\Big[
(\Delta_{1}^{2}+\Delta_{2}^{2}-\Delta_{3}^{2}+4\Delta_{3}-5)\ln\rho_{12}
%\nonumber\\
%&&\hspace{10pt}
+\;(\Delta_{1}^{2}+\Delta_{3}^{2}-\Delta_{2}^{2}-4\Delta_{3}+3)\ln\rho_{13}
\nonumber\\
&&\hspace{30pt}
+\;(\Delta_{3}^{2}+\Delta_{2}^{2}-\Delta_{1}^{2}-4\Delta_{3}+3)\ln\rho_{23}\Big]
%\nonumber\\
%&&\hspace{10pt}
-\gamma_E\left(\Delta_1+\Delta_2+\Delta_3-1\right)\Bigg\}.
\end{eqnarray}

Here \(\rho_{ij}\equiv z_{ij}\bar z_{ij}=|z_{ij}|^2\), and the expression
in braces is precisely the first-order deformation factor
\(\tfrac14 T+\Omega_1\), cf.~eq.~\eqref{AAA}.

In this section we develop an operator representation of the first
subleading inverse Mellin integral \eqref{I1_restate}. The
$b^2$-dependent conformal weights in the position-dependent Liouville
prefactor, together with the subleading contribution from the DOZZ
structure constant, generate polynomial dependence on the Mellin
parameters $\Delta_i$. Such polynomial insertions can be represented
by differential operators in the variables $\ln\omega_i$ acting on the
leading Mellin integral $\mathcal I^{(0)}$.

This provides a systematic evaluation of the first coefficient of the
Liouville deformation without modifying the Mellin kernel or
introducing a $b$-dependent integration measure. After evaluating the
required derivatives, $\mathcal I^{(1)}$ can be expressed in closed
form as a linear combination of the modified Bessel functions $K_1$
and $K_0$, with coefficient functions determined by the celestial
kinematics.

\subsubsection{Operator representation}
Let us introduce the logarithmic derivative operators
\begin{equation}\label{Ddef}
\mathcal D_i \equiv
-\frac{\partial}{\partial\ln\omega_i}
=
-\omega_i\frac{\partial}{\partial\omega_i},
\qquad i=1,2,3.
\end{equation}
Under the inverse Mellin transform, the factors
$\omega_i^{-\Delta_i}$ imply the correspondence
\begin{equation}
\Delta_i \longleftrightarrow \mathcal D_i,
\end{equation}
so that each polynomial insertion in the Mellin parameters is
mapped to the corresponding differential operator acting on the
inverse-Mellin result. Therefore the polynomial appearing in the
braces in \eqref{I1_restate} gives a linear differential operator
acting on $\mathcal I^{(0)}$. One finds\footnote{For a detailed
justification of eq.~\eqref{I1_operator}, see
appendix~\ref{app:operator_proof}.}
\begin{equation}\label{I1_operator}
\boxed{\quad
{\cal I}^{(1)}(\omega_i)
=
{\cal O}_{\rho,\gamma_E}\,
{\cal I}^{(0)}(\omega_i).
\quad}
\end{equation}
with
\begin{equation}\label{operator}
\begin{aligned}
{\cal O}_{\rho,\gamma_E}
&=
\frac{1}{4}\Big[
(\mathcal D_1^2+\mathcal D_2^2-\mathcal D_3^2+4\mathcal D_3-5)\ln\rho_{12}
\\[4pt]
&\qquad+
(\mathcal D_1^2+\mathcal D_3^2-\mathcal D_2^2-4\mathcal D_3+3)\ln\rho_{13}
\\[4pt]
&\qquad+
(\mathcal D_3^2+\mathcal D_2^2-\mathcal D_1^2-4\mathcal D_3+3)\ln\rho_{23}
\Big]
\\[6pt]
&\qquad
-\gamma_E
(\mathcal D_1+\mathcal D_2+\mathcal D_3-1).
\end{aligned}
\end{equation}

This identity provides the basic mechanism for generating the
subleading coefficients of the Mellin--Liouville construction from
the leading Mellin integral, without evaluating a new
three-dimensional Mellin integral at each order.

Since ${\cal O}_{\rho,\gamma_E}$ is linear and contains only
multiplicative functions of the celestial coordinates together with
logarithmic energy derivatives, it is sufficient to determine the
six building blocks
$\mathcal D_i\mathcal I^{(0)}$ and
$\mathcal D_i^2\mathcal I^{(0)}$,
$i=1,2,3$.

\subsubsection{Leading kernel and derivative structure}
We now turn to the computation of the derivatives
$\mathcal D_i\mathcal I^{(0)}$ and
$\mathcal D_i^2\mathcal I^{(0)}$.
To this end, we recall
\begin{equation}\label{Qdef}
Q^2
=
\omega_1\omega_3\rho_{13}
+
\omega_2\omega_3\rho_{23}
+
\omega_1\omega_2\rho_{12}.
\end{equation}

It is convenient to introduce
\begin{equation}\label{Qdefs}
{\cal P}
\equiv
\frac{4\omega_1\omega_2}{\omega_3\,M},
\qquad
{\cal F}(Q)
\equiv
\frac{1}{Q}
K_1\!\left(\frac{2Q}{M}\right),
\end{equation}
where $K_1$ denotes the modified Bessel function of the second kind.
Then
\begin{equation}\label{I0again}
\boxed{\quad
{\cal I}^{(0)}(\omega_i)
=
{\cal P}(\omega_i)\,{\cal F}(Q).
\quad}
\end{equation}

Let us introduce
\begin{equation}\label{Sdefs}
\begin{aligned}
S_1 &\equiv
Q^2-\rho_{23}\omega_2\omega_3
=
\omega_1(\rho_{13}\omega_3+\rho_{12}\omega_2),
\\[4pt]
S_2 &\equiv
Q^2-\rho_{13}\omega_1\omega_3
=
\omega_2(\rho_{23}\omega_3+\rho_{12}\omega_1),
\\[4pt]
S_3 &\equiv
Q^2-\rho_{12}\omega_1\omega_2
=
\omega_3(\rho_{13}\omega_1+\rho_{23}\omega_2).
\end{aligned}
\end{equation}
These satisfy
\begin{equation}
S_1+S_2+S_3=2Q^2,
\qquad
\omega_i\partial_{\omega_i}Q^2=S_i,
\qquad
\omega_i\partial_{\omega_i}Q=\frac{S_i}{2Q}.
\end{equation}

We also have
\begin{equation}\label{ellidef}
\mathcal D_i{\cal P}=\ell_i{\cal P},
\qquad
(\ell_1,\ell_2,\ell_3)=(-1,-1,+1),
\qquad
\mathcal D_i^2{\cal P}={\cal P}.
\end{equation}
Introducing
$x=2Q/M$, one finds
\begin{equation}
{\cal D}_i x
=
-\frac{S_i}{MQ}.
\end{equation}

We shall repeatedly use
\begin{equation}
\frac{{\rm d}}{{\rm d}x}K_0(x)=-K_1(x),
\qquad
\frac{{\rm d}}{{\rm d}x}K_1(x)
=
-K_0(x)-\frac{1}{x}K_1(x).
\end{equation}

\paragraph{First logarithmic derivative.}
Applying the product rule to \eqref{I0again}, one obtains
\begin{equation}
\mathcal D_i\mathcal I^{(0)}
=
\ell_i{\cal P}{\cal F}
+
{\cal P}\mathcal D_i{\cal F}.
\end{equation}
A straightforward calculation gives
\begin{equation}\label{eq:DF_final_app}
\mathcal D_i{\cal F}
=
\frac{S_i}{Q^3}K_1(x)
+
\frac{S_i}{MQ^2}K_0(x).
\end{equation}
Therefore
\begin{equation}
\boxed{
\mathcal D_i\mathcal I^{(0)}
=
\ell_i{\cal P}{\cal F}(Q)
+
{\cal P}
\left[
\frac{S_i}{Q^3}K_1\!\left(\frac{2Q}{M}\right)
+
\frac{S_i}{MQ^2}K_0\!\left(\frac{2Q}{M}\right)
\right].
}
\label{eq:DI0_final_app}
\end{equation}

\paragraph{Second logarithmic derivative.}
Differentiating once more yields
\begin{equation}
\mathcal D_i^2\mathcal I^{(0)}
=
{\cal P}{\cal F}
+
2\ell_i{\cal P}\mathcal D_i{\cal F}
+
{\cal P}\mathcal D_i^2{\cal F}.
\end{equation}
The remaining derivative can be evaluated using
\eqref{eq:DF_final_app}, giving
\begin{equation}\label{eq:D2F_final_app}
\mathcal D_i^2{\cal F}
=
\left[
-\frac{S_i}{Q^3}
+
\frac{2S_i^2}{Q^5}
+
\frac{S_i^2}{M^2Q^3}
\right]K_1(x)
+
\left[
-\frac{S_i}{MQ^2}
+
\frac{2S_i^2}{MQ^4}
\right]K_0(x).
\end{equation}
Hence
\begin{equation}
\boxed{
\begin{aligned}
\mathcal D_i^2\mathcal I^{(0)}
=
{\cal P}\Bigg\{
&{\cal F}(Q)
+
\Bigg[
\frac{(2\ell_i-1)S_i}{Q^3}
+
\frac{2S_i^2}{Q^5}
+
\frac{S_i^2}{M^2Q^3}
\Bigg]
K_1\!\left(\frac{2Q}{M}\right)
\\[5pt]
&+
\Bigg[
\frac{(2\ell_i-1)S_i}{MQ^2}
+
\frac{2S_i^2}{MQ^4}
\Bigg]
K_0\!\left(\frac{2Q}{M}\right)
\Bigg\}.
\end{aligned}
}
\label{D2I0}
\end{equation}

Equations \eqref{eq:DI0_final_app} and \eqref{D2I0} provide the
building blocks required to evaluate the first subleading coefficient
of the Liouville deformation.

\subsubsection{Exact Bessel decomposition}
We now insert \eqref{eq:DI0_final_app} and \eqref{D2I0} into the
operator \eqref{operator}. Since the operator is linear, the
contribution proportional to $K_1$ and that proportional to $K_0$ can
be collected separately. 

The coefficient of \(K_0(2Q/M)\) is
\begin{equation}
{\sf C}_0(\omega_i;\rho_{jk})
=\frac{1}{4}\left[
\ln\rho_{12}\,\Xi_{12}^{(0)}
+\ln\rho_{13}\,\Xi_{13}^{(0)}
+\ln\rho_{23}\,\Xi_{23}^{(0)}
\right]
-\frac{2\gamma_E}{M},
\end{equation}
with \(\Xi_{ij}^{(0)}\) given by
\begin{eqnarray}
\label{Xi12_0}
\Xi_{12}^{(0)}
&=&
-\frac{
Q^4-Q^2(S_1+S_2)+2S_1S_2
}{2MQ^4},\\
\label{Xi13_0}
\Xi_{13}^{(0)}
&=&
\frac{
Q^4-4Q^2S_1-Q^2S_2+2S_1^2+2S_1S_2
}{2MQ^4},\\
\label{Xi23_0}
\Xi_{23}^{(0)}
&=&
\frac{
Q^4-Q^2S_1-4Q^2S_2+2S_1S_2+2S_2^2
}{2MQ^4}.
\end{eqnarray} 
The \(-2\gamma_E/M\) term comes from
\begin{equation}
-\gamma_E(\mathcal D_1+\mathcal D_2+\mathcal D_3-1)\mathcal I^{(0)}
=
-\frac{2\gamma_E\mathcal P}{M}
K_0\!\left(\frac{2Q}{M}\right),
\end{equation}
which follows from
\begin{equation}
(\mathcal D_1+\mathcal D_2+\mathcal D_3)\mathcal I^{(0)}
=
\mathcal I^{(0)}
+
\frac{2\mathcal P}{M}
K_0\!\left(\frac{2Q}{M}\right).
\end{equation}
Similarly, the coefficient of \(K_1(2Q/M)\) is
\begin{equation}
{\sf C}_1(\omega_i;\rho_{jk})
=
\frac{1}{4}\Big[
\ln\rho_{12}\,\Xi_{12}^{(1)}
+\ln\rho_{13}\,\Xi_{13}^{(1)}
+\ln\rho_{23}\,\Xi_{23}^{(1)}
\Big],
\end{equation}
with \(\Xi_{ij}^{(1)}\) given by
\begin{subequations}\label{Xi_group}
\begin{align}
\Xi_{12}^{(1)}
&=
-\frac{1}{2M^2Q^5}\Big[
M^2Q^4 - M^2Q^2(S_1+S_2) + 2M^2S_1S_2
\nonumber\\
&\hspace{4.5em}
+\, 2Q^6 - 2Q^4(S_1+S_2) + Q^2S_1S_2
\Big], \label{Xi12}\\[4pt]
\Xi_{13}^{(1)}
&=
\frac{1}{2M^2Q^5}\Big[
M^2Q^4 - 4M^2Q^2S_1 - M^2Q^2S_2 + 2M^2S_1^2 + 2M^2S_1S_2
\nonumber\\
&\hspace{4.5em}
+\, 2Q^6 - 2Q^4(S_1+S_2) + Q^2S_1^2 + Q^2S_1S_2
\Big], \label{Xi13}\\[4pt]
\Xi_{23}^{(1)}
&=
\frac{1}{2M^2Q^5}\Big[
M^2Q^4 - M^2Q^2S_1 - 4M^2Q^2S_2 + 2M^2S_1S_2 + 2M^2S_2^2
\nonumber\\
&\hspace{4.5em}
+\, 2Q^6 - 2Q^4(S_1+S_2) + Q^2S_1S_2 + Q^2S_2^2
\Big]. \label{Xi23}
\end{align}
\end{subequations}
Combining the two pieces yields the exact decomposition
\begin{equation}\label{I1_Bessel_decomp}\boxed{\;
\mathcal I^{(1)}(\omega_i)
=
\mathcal P(\omega_i)\left[
{\sf C}_1(\omega_i;\rho_{jk})\,K_1\!\left(\frac{2Q}{M}\right)
+
{\sf C}_0(\omega_i;\rho_{jk})\,K_0\!\left(\frac{2Q}{M}\right)
\right].\;}
\end{equation}

Using
\begin{equation}
{\cal I}^{(0)}(\omega_i)
=
{\cal P}(\omega_i)
\frac{K_1\!\left(\frac{2Q}{M}\right)}{Q},
\end{equation}
the contribution to the first deformation coefficient may equivalently be written as
\begin{equation}\label{I1vsI0}
\boxed{
{\cal I}^{(1)}(\omega_i)
=
{\cal I}^{(0)}(\omega_i)
\left[
Q\,{\sf C}_1(\omega_i;\rho_{jk})
+
Q\,{\sf C}_0(\omega_i;\rho_{jk})
\frac{K_0\!\left(\frac{2Q}{M}\right)}
{K_1\!\left(\frac{2Q}{M}\right)}
\right].
}
\end{equation}

This representation makes explicit that the first subleading
coefficient is generated entirely from the leading Mellin kernel by
the differential operator \eqref{operator}. No new Mellin measure or
independent Mellin kernel is introduced.

\subsection{Small-$Q$ behaviour of the first Liouville deformation}
\label{subleading-smallQ}
We next examine the small-$Q$ behaviour of the first deformation
coefficient. Combining \eqref{A0_repeat}, \eqref{A1_restate} and \eqref{I1vsI0}, 
one obtains
\begin{equation}\label{A1vsA0}\boxed{\;
\delta{\cal A}_{\rm 3pt, L}^{(1)}
=
{\cal A}_{\rm 3pt}^{\rm tree}
\left[
Q\,{\sf C}_1(\omega_i;\rho_{jk})
+
Q\,{\sf C}_0(\omega_i;\rho_{jk})
\frac{K_0\!\left(\frac{2Q}{M}\right)}
{K_1\!\left(\frac{2Q}{M}\right)}
\right],\;}
\end{equation}
where ${\cal A}_{\rm 3pt}^{\rm tree}\equiv{\cal A}_{\rm 3pt}^{(0)}$ denotes the leading
tree-level Yang--Mills amplitude (\ref{A0f}).

\paragraph{Momentum-space kinematics.}
To make contact with standard field-theoretic variables, it is
convenient to introduce
\begin{equation}
s_{ij}\equiv2p_i\!\cdot p_j,
\qquad
\rho_{ij}=\frac{s_{ij}}{\omega_i\omega_j},
\qquad
Q^2=s_{12}+s_{13}+s_{23},
\end{equation}
and
\begin{equation}
S_1=s_{12}+s_{13},
\qquad
S_2=s_{12}+s_{23},
\qquad
S_3=s_{13}+s_{23}.
\end{equation}
We further define
\begin{equation}
s_1=\frac{S_1}{Q^2},
\qquad
s_2=\frac{S_2}{Q^2},
\qquad
s_3=\frac{S_3}{Q^2},
\qquad
s_1+s_2+s_3=2.
\end{equation}

\paragraph{Small-$Q$ behaviour.}
In the limit $Q\to0$ we keep the dimensionless ratios $s_i$ fixed,
so that $S_i=Q^2s_i$. The coefficient \(Q{\sf C}_1(\omega_i;\rho_{jk})\)
has a finite limit:
\begin{eqnarray}\label{F_explicit}
Q\,{\sf C}_1(\omega_i;\rho_{jk})
\;\longrightarrow\;
\mathfrak{F}(\rho_{ij};s_1,s_2)
&=&
\frac18\Big[
-\ln\rho_{12}\,
(1-s_1-s_2+2s_1s_2)
\nonumber\\
&&+
\ln\rho_{13}\,
(1-4s_1-s_2+2s_1^2+2s_1s_2)
\nonumber\\
&&+
\ln\rho_{23}\,
(1-s_1-4s_2+2s_1s_2+2s_2^2)
\Big],
\end{eqnarray}
where \(\mathfrak{F}(\rho_{ij};s_1,s_2)\) is a kinematic function built 
from logarithms of \(\rho_{ij}\) and the ratios \(s_i\)
(cf.~appendix~\ref{app:soft_limit_C1}).

For the second contribution, 
\begin{equation}
Q\,{\sf C}_0(\omega_i;\rho_{jk})\,
\frac{K_0\!\left(\frac{2Q}{M}\right)}{K_1\!\left(\frac{2Q}{M}\right)},
\end{equation}
the apparent inverse powers of $Q$
cancel after substituting $S_i=Q^2s_i$, so that
${\sf C}_0(\omega_i;\rho_{jk})$ remains finite as $Q\to0$.

For the Bessel factor, let \(x=2Q/M\). Using the standard small-\(x\) expansions
\begin{equation}
K_0(x)=-\ln\!\left(\frac{x}{2}\right)-\gamma_E+{\cal O}(x^2),
\qquad
K_1(x)=\frac{1}{x}+{\cal O}(x\ln x),
\end{equation}
we find
\begin{equation}
\frac{K_0(x)}{K_1(x)}
=
\Big(-\ln\!\left(\frac{x}{2}\right)-\gamma_E+{\cal O}(x^2)\Big)
\Big(x+{\cal O}(x^3\ln x)\Big)
=
{\cal O}(x\ln x),
\end{equation}
and therefore\footnote{The term \({\cal O}(x^3 \ln x)\) comes from the correction to \(1/K_1(x)\), but it is subleading.}
\begin{equation}
\frac{K_0\!\left(\frac{2Q}{M}\right)}{K_1\!\left(\frac{2Q}{M}\right)}
=
{\cal O}(Q\ln Q).
\end{equation}
Multiplying by the extra factor of \(Q\) and
combining this with the finiteness of \({\sf C}_0(\omega_i;\rho_{jk})\), one concludes that
\begin{equation}
Q\,{\sf C}_0(\omega_i;\rho_{jk})\,
\frac{K_0\!\left(\frac{2Q}{M}\right)}{K_1\!\left(\frac{2Q}{M}\right)}
=
{\cal O}\!\big(Q^2\ln Q\big).
\end{equation}

Hence the first deformation coefficient has the finite limiting
ratio
\begin{equation}\boxed{\;
\frac{
\delta{\cal A}_{\rm 3pt, L}^{(1)}
}{
{\cal A}_{\rm 3pt}^{\rm tree}
}
=
\mathfrak{F}(\rho_{ij};s_1,s_2)
+
{\cal O}(Q^2\ln Q).\;}
\end{equation}
The leading small-$Q$ scaling therefore follows that of the
tree-level amplitude, while the first Liouville deformation
adds a finite logarithmic kinematic dressing together with
subleading nonanalytic terms in $Q$.

The logarithms appearing in this expansion arise from the
small-argument behaviour of the modified Bessel functions and, within
the present fixed-$\sigma$ construction, appear as nonanalytic terms
in the small-$Q$ expansion of the Liouville deformation. Their possible
relation to infrared singularities of genuine Yang--Mills loop
amplitudes cannot be assessed within this parametrization alone and
requires a reformulation in the STZ parametrization.

\section{Relation to the STZ parametrization}
\label{sec:STZ_loop}
The perturbative construction developed above is formulated at fixed
Liouville parameters $\sigma_\pm$. It is useful to relate this
organization to the finite-$b$ parametrization introduced in
\cite{Stieberger2023b}. The purpose of this section is not to
perform the full one-loop comparison, but to make the relation between
the two descriptions more explicit.

\subsection{Fixed-${\sigma}$ and STZ parametrizations}
For a gluon of helicity $\pm$, the standard celestial conformal weights
are
\begin{equation}
\label{celestial_weights}
h_\pm^{\rm cel}(\Delta)
=\frac{\Delta\mp1}{2}.
\end{equation}
Recall that the light Liouville operator is parametrized as
\begin{equation}\label{hL}
{\sf V}_{h(\sigma)}={\rm e}^{2b\sigma\phi},
\qquad
h(\sigma)=\sigma+b^2\sigma(1-\sigma).
\end{equation}

In the fixed-${\sigma}$ prescription used in this work, the
Liouville parameters $\sigma_\pm$ are identified with the celestial weights,
\begin{equation}
\label{fixed_sigma_relation}
\sigma_\pm
=
\frac{\Delta\mp1}{2},
\end{equation}
and are kept fixed as $b$ is varied. Consequently,
\begin{equation}
\label{fixed_sigma_weight}
h(\sigma_\pm)
=
\frac{\Delta\mp1}{2}
+
b^2\sigma_\pm(1-\sigma_\pm).
\end{equation}

The STZ prescription implements the finite-$b$ correction differently.
There the celestial conformal weight is kept fixed, while the Liouville
parameter is allowed to depend on $b$:
\begin{equation}
\label{STZ_condition}
h\!\left(\sigma_\pm^{\rm STZ}\right)
=
\frac{\widehat\Delta\mp1}{2}.
\end{equation}
Thus, for fixed $\widehat\Delta$, the parameter
$\sigma_\pm^{\rm STZ}$ acquires a finite-$b$ dependence.

To compare the two organizations, it is convenient to consider the same
value of the Liouville parameter $\sigma_\pm$. In the fixed-$\sigma$
description one has
\begin{equation}
\label{Delta_sigma}
\Delta
=
2\sigma_\pm\pm1,
\end{equation}
whereas in the STZ parametrization
\begin{equation}
\label{Deltahat_sigma}
\widehat\Delta
=
2h(\sigma_\pm)\pm1.
\end{equation}
Using \eqref{hL}, we therefore obtain
\begin{align}
\widehat\Delta
&=
2\sigma_\pm
+
2b^2\sigma_\pm(1-\sigma_\pm)
\pm1
\nonumber\\
&=
\Delta
+
2b^2\sigma_\pm(1-\sigma_\pm).
\end{align}
Hence
\begin{equation}
\label{Delta_map_final}
\boxed{
\widehat\Delta
=
\Delta
+
2b^2\sigma_\pm(1-\sigma_\pm).
}
\end{equation}

Equation \eqref{Delta_map_final} is a relation between the two
parametrizations at a common value of $\sigma_\pm$. It should not be
viewed as the defining equation of the STZ prescription: in the latter,
one may instead regard $\widehat\Delta$ as fixed and determine the
corresponding $b$-dependent value of $\sigma_\pm^{\rm STZ}$.
This can be seen explicitly by expanding around the semiclassical
value. Consider, for example, one of the two helicity sectors and write
\begin{equation}
\label{sigma_STZ_expansion}
\sigma^{\rm STZ}(b)
=
\sigma_0+b^2\sigma_1+\mathcal O(b^4),
\end{equation}
where
\begin{equation}
\sigma_0
=
\frac{\widehat\Delta\mp1}{2}.
\end{equation}
Substitution into \eqref{hL} gives
\begin{align}
h(\sigma^{\rm STZ}(b))
&=
\sigma_0
+b^2\sigma_1
+b^2\sigma_0(1-\sigma_0)
+\mathcal O(b^4).
\end{align}
Requiring the conformal weight to remain equal to its semiclassical value
$\sigma_0$ gives
\begin{equation}
\sigma_1+\sigma_0(1-\sigma_0)=0,
\end{equation}
and hence
\begin{equation}
\label{sigma_STZ_bdependence}
\sigma^{\rm STZ}(b)
=
\sigma_0
-b^2\sigma_0(1-\sigma_0)
+\mathcal O(b^4).
\end{equation}
Thus the finite-$b$ correction to the Liouville conformal weight is
compensated by a corresponding $b$-dependence of the Liouville
parameter.

\subsection{Induced transformation of the Mellin kernel}
\label{sec:induced_Mellin_transformation}

For a given helicity sector, let us define
\begin{equation}
f_\pm(\Delta)
=
2\sigma_\pm(\Delta)
\bigl[1-\sigma_\pm(\Delta)\bigr],
\end{equation}
with $\sigma_\pm(\Delta)$ given by \eqref{fixed_sigma_relation}.
Explicitly,
\begin{equation}
f_+(\Delta)
=
\frac{(\Delta-1)(3-\Delta)}{2},
\qquad
f_-(\Delta)
=
\frac{1-\Delta^2}{2}.
\end{equation}
The relation between the two Mellin parameters obtained in the
previous subsection can then be written as
\begin{equation}
\widehat\Delta
=
\Delta+b^2f_\pm(\Delta).
\label{eq:Delta_reparam}
\end{equation}
Here $\widehat{\Delta}$ denotes the Mellin parameter used in the STZ
parametrization, while $\Delta$ denotes the parameter used in the
fixed-$\sigma$ prescription.

When the STZ representation is rewritten in terms of the variable
$\Delta$, three elementary transformations are induced. First, the
integration measure becomes
\begin{equation}
{\rm d}\widehat\Delta
=
\left[
1+b^2f_\pm'(\Delta)+O(b^4)
\right]
{\rm d}\Delta.
\label{eq:Jacobian_reparam}
\end{equation}
Equivalently, the inverse Jacobian is
\begin{equation}
\frac{{\rm d}\Delta}{{\rm d}\widehat\Delta}
=
1-b^2f_\pm'(\widehat\Delta)
+O(b^4).
\end{equation}
The higher powers of $b^2$ generated by the inverse Jacobian arise
simply from inverting the nonlinear relation (\ref{eq:Delta_reparam}).

For completeness, since $f_\pm$ is quadratic, the inverse
transformation may also be written in closed form. Introducing
$\epsilon_+=1$ and $\epsilon_-=-1$, so that
$\sigma_\pm=(\Delta-\epsilon_\pm)/2$, one obtains
\begin{equation}
\Delta_\pm(\widehat\Delta)
=
\epsilon_\pm+
\frac{
1+b^2-
\sqrt{(1+b^2)^2-2b^2(\widehat\Delta-\epsilon_\pm)}
}{b^2},
\end{equation}
where the branch continuously connected to
$\Delta_\pm=\widehat\Delta$ as $b\to0$ is understood. The corresponding
exact inverse Jacobian is
\begin{equation}
\frac{{\rm d}\Delta_\pm}{{\rm d}\widehat\Delta}
=
\frac{1}{
\sqrt{(1+b^2)^2-2b^2(\widehat\Delta-\epsilon_\pm)}
}.
\end{equation}
Its expansion reproduces the perturbative result above.

Second, the Mellin kernel transforms as
\begin{align}
\left(\frac{\omega}{M}\right)^{-\widehat\Delta}
&=
\left(\frac{\omega}{M}\right)^{-\Delta}
\exp\!\left[
-b^2f_\pm(\Delta)\ln\frac{\omega}{M}
\right]
\nonumber\\
&=
\left(\frac{\omega}{M}\right)^{-\Delta}
\left[
1-b^2f_\pm(\Delta)\ln\frac{\omega}{M}+O(b^4)
\right].
\label{eq:kernel_reparam}
\end{align}

Third, the Mellin-space function itself is evaluated at the shifted
argument,
\begin{equation}
\widehat{\cal F}_b(\widehat\Delta)
=
\widehat{\cal F}_b(\Delta)
+
b^2f_\pm(\Delta)\,
\partial_\Delta\widehat{\cal F}_b(\Delta)
+
O(b^4).
\label{eq:F_reparam}
\end{equation}

Combining the three contributions, the change of Mellin variable
induces the differential operator
\begin{equation}
\boxed{
{\cal R}
=
f_\pm(\Delta)\partial_\Delta
+
f_\pm'(\Delta)
-
f_\pm(\Delta)\ln\frac{\omega}{M}.
}
\label{eq:R_single_leg}
\end{equation}
For the three-gluon amplitude with helicities $(-,-,+)$ this becomes
\begin{equation}
\boxed{
{\cal R}
=
\sum_{i=1}^3
\left[
f_i(\Delta_i)\partial_{\Delta_i}
+
f_i'(\Delta_i)
-
f_i(\Delta_i)
\ln\frac{\omega_i}{M}
\right],
}
\label{eq:R_final}
\end{equation}
where
\begin{equation}
f_1=f_-(\Delta_1),
\qquad
f_2=f_-(\Delta_2),
\qquad
f_3=f_+(\Delta_3).
\end{equation}

The three terms in each summand have a direct origin: the first comes
from the shift of the argument of the Mellin-space function, the second
from the Jacobian of the change of variables, and the third from the
transformation of the Mellin kernel. Thus, ${\cal R}$ is a kinematical
operator induced by the Mellin--Liouville reparametrization rather than
an additional dynamical interaction.

The operator may alternatively be represented in terms of derivatives
with respect to the external energies. Introducing
\begin{equation}
{\cal D}_i=-\omega_i\partial_{\omega_i},
\qquad
{\cal D}_i\omega_i^{-\Delta_i}
=
\Delta_i\omega_i^{-\Delta_i},
\end{equation}
one may write
\begin{equation}
f_-({\cal D})
=
\frac{1-{\cal D}^2}{2},
\qquad
f_+({\cal D})
=
\frac{({\cal D}-1)(3-{\cal D})}{2}.
\end{equation}
Since
\begin{equation}
[{\cal D}_i,\ln(\omega_i/M)]=-1,
\end{equation}
the ordering of the corresponding energy-space operators has to be
kept fixed.

\subsection{Reparametrization of the inverse Mellin integral}
\label{sec:reparam_inverse_Mellin}

We can now combine the elementary transformations described above at
the level of the complete inverse Mellin representation. Schematically,
the STZ three-point amplitude can be written as
\begin{equation}
{\cal A}_{\rm 3pt}^{\rm STZ}
=
\frac{1}{(2\pi i)^3}
\int
\prod_{i=1}^3 {\rm d}\widehat{\Delta}_i\,
M^{-3}
\prod_{i=1}^3
\left(\frac{\omega_i}{M}\right)^{-\widehat{\Delta}_i}
\,
\widehat{\cal F}_b
(\widehat{\Delta}_1,\widehat{\Delta}_2,\widehat{\Delta}_3).
\label{eq:STZ_reparam_start}
\end{equation}
For the $(-,-,+)$ configuration, we use
\begin{equation}
f_1=f_-(\Delta_1),
\qquad
f_2=f_-(\Delta_2),
\qquad
f_3=f_+(\Delta_3).
\end{equation}

Applying the change of variables
$\widehat{\Delta}_i=\Delta_i+b^2f_i(\Delta_i)$ gives
\begin{equation}
{\rm d}\widehat{\Delta}_i
=
\left[
1+b^2f_i'(\Delta_i)+O(b^4)
\right]{\rm d}\Delta_i,
\end{equation}
and
\begin{equation}
\left(\frac{\omega_i}{M}\right)^{-\widehat{\Delta}_i}
=
\left(\frac{\omega_i}{M}\right)^{-\Delta_i}
\left[
1-b^2f_i(\Delta_i)L_i+O(b^4)
\right],
\qquad
L_i=\ln\frac{\omega_i}{M}.
\end{equation}
Moreover,
\begin{equation}
\widehat{\cal F}_b(\widehat{\Delta}_1,\widehat{\Delta}_2,\widehat{\Delta}_3)
=
\widehat{\cal F}_b(\Delta_1,\Delta_2,\Delta_3)
+
b^2\sum_{i=1}^3
f_i(\Delta_i)
\partial_{\Delta_i}\widehat{\cal F}_b
+
O(b^4).
\end{equation}

At the level of the perturbative Mellin integrand, and suppressing the
corresponding deformation of the integration contours, these
transformations combine into
\begin{equation}
\begin{split}
{\cal A}_{\rm 3pt}^{\rm STZ}
=
\frac{1}{(2\pi i)^3}
\int
\prod_{i=1}^3 {\rm d}\Delta_i\,
M^{-3}
\prod_{i=1}^3
\left(\frac{\omega_i}{M}\right)^{-\Delta_i}
\Big[
\widehat{\cal F}_b
+
b^2{\cal R}\widehat{\cal F}_b
+
O(b^4)
\Big],
\end{split}
\label{eq:R_reparam_result}
\end{equation}
where ${\cal R}$ is given in \eqref{eq:R_final}.

This result has a simple interpretation. The change from
$\widehat{\Delta}$ to $\Delta$ does not merely relabel the Mellin
integration variable. It induces a definite differential
reorganization of the Mellin-space integrand through the operator
${\cal R}$. The operator contains only the effects associated with the
Mellin--Liouville reparametrization itself.

To make contact with the perturbative expansion, let us write
\begin{equation}
\widehat{\cal F}_b
=
\widehat{\cal F}^{(0)}
+
b^2\widehat{\cal F}^{(1)}
+
O(b^4).
\label{eq:F_STZ_final}
\end{equation}
Then the integrand in the reparametrized STZ representation takes the
form
\begin{equation}
\widehat{\cal F}^{(0)}
+
b^2
\left[
\widehat{\cal F}^{(1)}
+
{\cal R}\widehat{\cal F}^{(0)}
\right]
+
O(b^4).
\label{eq:effective_first_order}
\end{equation}
Here the first term at order $b^2$ represents the intrinsic
subleading contribution in the STZ parametrization, while the second
is generated by the change of Mellin variable.

At leading order, the two parametrizations coincide since
$\widehat{\Delta}=\Delta+O(b^2)$ and both reduce to the standard
celestial three-gluon representation. Accordingly, one may identify
the leading Mellin-space quantities,
\begin{equation}
\widehat{\cal F}^{(0)}
=
{\cal F}^{(0)},
\end{equation}
where ${\cal F}^{(0)}$ denotes the leading fixed-$\sigma$ expression.

At subleading order, however, the relation is more involved. The
coefficient ${\cal F}^{(1)}_L$ obtained from the fixed-$\sigma$
Liouville expansion cannot be identified directly with
$\widehat{\cal F}^{(1)}$ or with ${\cal R}\widehat{\cal F}^{(0)}$
separately. Rather, establishing the relation between the two
perturbative organizations requires a comparison of the complete
Mellin-space integrands in the two parametrizations, including their
respective finite-$b$ Liouville contributions and the
Mellin--Liouville reparametrization.

This observation is also relevant for the possible loop interpretation
of the finite-$b$ expansion in the STZ framework. The operator
${\cal R}$ by itself should not be identified with a Yang--Mills loop
correction. Instead, it specifies the kinematical contribution induced
when the STZ Mellin representation is rewritten in the fixed-$\sigma$
variable. A comparison with the one-loop celestial Yang--Mills
amplitude therefore requires the simultaneous evaluation of the
intrinsic subleading term in the STZ parametrization and of the
reparametrization contribution generated by ${\cal R}$, followed by
the inverse Mellin transformation.

Carrying out this comparison lies beyond the scope of the present
work. The purpose of the analysis here is to establish the finite-$b$
Mellin--Liouville map and the induced operatorial transformation of the
inverse Mellin representation. A detailed comparison of the resulting
subleading terms in the fixed-$\sigma$ and STZ parametrizations,
including their possible relation to the one-loop celestial
Yang--Mills amplitude, will be pursued separately.

%%%%%%%%%%%%%%%%%%%%%%%%%%%%%%%%%%%%%%%%%%%%%%%%%%%%%%%%%%%%%%%%%%%%%%

\section{Concluding remarks}
\label{sec:concluding_remarks}

In this work we have formulated a perturbative Mellin--Liouville
prescription for the celestial three-gluon three-point representation.
The essential feature of the construction is that the light Liouville
parameters $\sigma_i$ are kept fixed in the small-$b$ expansion. The
corresponding Liouville conformal weights then acquire their exact
$b$-dependent corrections. At finite $b$, the resulting object is
therefore understood as a controlled Liouville deformation of the
standard celestial three-gluon construction, rather than as a standard
finite-$b$ celestial CFT correlator with the usual celestial weights.

At $b=0$, the Liouville-dressed conformal weights reduce to their
standard celestial values, and the inverse Mellin representation
reduces to the construction of Stieberger, Taylor and Zhu
\cite{Stieberger2023b}. Its evaluation reproduces the corresponding
tree-level three-gluon Yang--Mills amplitude in the appropriate
small-total-momentum limit. This provides the leading-order
consistency check of the present prescription.

The main technical result is the explicit evaluation of the first
Liouville deformation. At order $b^2$, the corrections to the
Liouville conformal weights and to the light DOZZ structure constant
produce polynomial insertions in the Mellin variables. These insertions
can be converted into differential operators acting on the leading
Mellin integral. The resulting inverse Mellin integral can be evaluated
analytically in terms of the modified Bessel functions $K_1$ and
$K_0$. The small-$Q$ expansion also exhibits logarithmic
nonanalyticities of the form $Q^2\ln Q$, together with logarithms of
the celestial distances.

We have also clarified the relation between the fixed-$\sigma$ and STZ
finite-$b$ parametrizations. They agree in the semiclassical limit but
organize the finite-$b$ deformation differently. The two Mellin
parameters are related by a nonlinear reparametrization (\ref{eq:Delta_reparam}),
which induces corresponding transformations of the Mellin measure, the
Mellin kernel and the argument of the Mellin-space function. These
contributions can be combined into the differential operator
${\cal R}$ derived in section~\ref{sec:induced_Mellin_transformation}.
This provides an explicit operatorial map between the two
parametrizations at the level of the inverse Mellin representation.

The present analysis does not identify the $b^2$ coefficient with a
Yang--Mills loop amplitude. In particular, a complete comparison with
the STZ loop interpretation requires a simultaneous treatment of the
intrinsic subleading contribution in the STZ parametrization and the
contribution induced by the Mellin reparametrization operator
${\cal R}$. This requires a parallel calculation of the corresponding
Mellin-space integrands and is beyond the scope of the present work.

Several extensions are natural. The first is to carry out this
comparison explicitly and determine whether, and to what extent, the
Liouville deformation can be related to genuine loop-level structures
in celestial Yang--Mills theory. A second direction is the systematic
calculation of higher orders in $b^2$. The polynomial structure of the
light DOZZ expansion suggests that the operator method developed here
can be extended recursively to higher orders. Finally, it would be
interesting to investigate whether the same prescription can be
extended to four-point celestial amplitudes, where conformal blocks
and factorization provide additional nontrivial tests.

The three-gluon case therefore provides a simple analytic setting in
which a finite-$b$ Liouville deformation of the standard celestial
construction can be defined, evaluated perturbatively, and compared
between different Mellin parametrizations. The results presented here
establish this framework while leaving the detailed relation to
celestial Yang--Mills loop amplitudes for future work.

\appendix
\renewcommand{\theequation}{\thesection.\arabic{equation}}
\setcounter{equation}{0}

\section{Proof of the operator representation}
\label{app:operator_proof}
In this appendix we justify the identity
\[
\mathcal I^{(1)}(\omega_i)=\mathcal O_{\rho,\gamma_E}\,\mathcal I^{(0)}(\omega_i),
\]
which converts the polynomial dependence on the Mellin variables
$\Delta_i$ in \eqref{I1_restate} into differential operators acting on the leading
inverse Mellin integral $\mathcal I^{(0)}$.

\paragraph{Setup.}
Let
\begin{equation*}
\mathcal I^{(0)}(\omega_1,\omega_2,\omega_3)
=
\left(\frac{1}{2\pi i}\right)^3
\int_{c_1-i\infty}^{c_1+i\infty}\! {\rm d}\Delta_1
\int_{c_2-i\infty}^{c_2+i\infty}\! {\rm d}\Delta_2
\int_{c_3-i\infty}^{c_3+i\infty}\! {\rm d}\Delta_3\;
\mathcal K(\Delta_1,\Delta_2,\Delta_3;\omega_i),
\end{equation*}
where
\begin{equation*}
\mathcal K(\Delta_i;\omega_i)
=
M^{\Delta_1+\Delta_2+\Delta_3-1}
\omega_1^{-\Delta_1}\omega_2^{-\Delta_2}\omega_3^{-\Delta_3}
{\sf F}(\Delta_1,\Delta_2,\Delta_3),
\end{equation*}
and ${\sf F}$ denotes the remaining part of the Mellin integrand, namely the product
of gamma functions and all factors independent of the external variables
$\omega_i$. Using the same notation as in the main text, we define
\begin{equation*}
\mathcal D_i \equiv -\frac{\partial}{\partial \ln \omega_i}
= -\omega_i\frac{\partial}{\partial \omega_i},
\qquad i=1,2,3.
\end{equation*}

The key observation is that $\mathcal D_i$ acts only on the external variable
$\omega_i$ and leaves the Mellin variables $\Delta_j$ untouched. Since
\[
\mathcal D_i\,\omega_i^{-\Delta_i}=\Delta_i\,\omega_i^{-\Delta_i},
\]
we expect that, under the inverse Mellin integral, multiplication by $\Delta_i$
may be replaced by $\mathcal D_i$.

\paragraph{Assumptions.}
The argument uses the standard theorem on differentiation under the integral
sign (Leibniz rule), justified here by absolute convergence and a uniform
integrable bound on the vertical contours. Concretely, we assume:
\begin{itemize}
\item[1.] The integrand $\mathcal K(\Delta_i;\omega_i)$ is holomorphic in each
$\Delta_i$ in a strip containing the vertical contours
$\Re\Delta_i=c_i$.
\item[2.] For $\omega_i$ in any compact subset of $(0,\infty)^3$, the integral
converges absolutely and uniformly on the contours after multiplication by
any polynomial in $\Delta_i$.
\item[3.] The Gamma function factors in ${\sf F}(\Delta_1,\Delta_2,\Delta_3)$ provide
exponential decay in the imaginary directions. This follows from Stirling's
formula on vertical lines:
\[
\Gamma(\sigma+it)
=
\sqrt{2\pi}\,|t|^{\sigma-\frac12}{\rm e}^{-\frac{\pi}{2}|t|}
\bigl(1+{\cal O}(|t|^{-1})\bigr),
\qquad |t|\to\infty,
\]
uniformly for $\sigma$ in compact sets. Hence products of Gamma functions
dominate any polynomial growth coming from powers of $\Delta_i$.
\end{itemize}
These assumptions are standard in Mellin transform theory and are satisfied in the present case.

\paragraph{Action of $\mathcal D_i$ under the integral.}
\begin{lemma}
Under the assumptions above,
\begin{equation*}
\mathcal D_i\,\mathcal I^{(0)}(\omega_1,\omega_2,\omega_3)
=
\left(\frac{1}{2\pi i}\right)^3
\int {\rm d}^3\Delta\;
\Delta_i\,\mathcal K(\Delta_1,\Delta_2,\Delta_3;\omega_i).
\end{equation*}
More generally, for any nonnegative integer $n$,
\begin{equation*}
\mathcal D_i^{\,n}\,\mathcal I^{(0)}
=
\left(\frac{1}{2\pi i}\right)^3
\int {\rm d}^3\Delta\;
\Delta_i^{n}\,\mathcal K(\Delta_i;\omega_i).
\end{equation*}
\end{lemma}
{\it Proof.}
We first note that
\[
\mathcal D_i\,\mathcal K(\Delta_j;\omega_j)
=
-\frac{\partial}{\partial\ln\omega_i}
\left(
M^{\sum_j\Delta_j-1}\prod_{j=1}^3 \omega_j^{-\Delta_j}{\sf F}(\Delta_j)
\right)
=
\Delta_i\,\mathcal K(\Delta_j;\omega_j),
\]
because the only $\omega_i$-dependence is contained in the factor
$\omega_i^{-\Delta_i}$.

By the assumptions stated above, the derivative with respect to
$\ln\omega_i$ may be passed under the integral sign. Indeed, the integrand and
its first $\omega_i$-derivative are dominated by an integrable function on the
vertical contours, thanks to the exponential decay from the Gamma functions.
Therefore the Leibniz integral rule applies, giving
\[
\mathcal D_i \mathcal I^{(0)}
=
\left(\frac{1}{2\pi i}\right)^3
\int {\rm d}^3\Delta\;
\mathcal D_i \mathcal K
=
\left(\frac{1}{2\pi i}\right)^3
\int {\rm d}^3\Delta\;
\Delta_i\,\mathcal K.
\]
Iterating the same argument yields the higher-power identity
for $\mathcal D_i^{\,n}$.

\paragraph{Polynomial replacement rule.}
Since the operators $\mathcal D_i$ commute with one another,
any polynomial ${\sf P}(\Delta_1,\Delta_2,\Delta_3)$ may be replaced by the same
polynomial in the differential operators:
\begin{equation*}
{\sf P}(\mathcal D_1,\mathcal D_2,\mathcal D_3)\,\mathcal I^{(0)}
=\left(\frac{1}{2\pi i}\right)^3
\int {\rm d}^3\Delta\;
{\sf P}(\Delta_1,\Delta_2,\Delta_3)\,\mathcal K(\Delta_i;\omega_i).
\end{equation*}
This is the precise sense in which one may ``factor out'' the differential
operators from the inverse Mellin integral: they do not literally commute past
the integral as constants, but rather act on the external variables in such a
way that their action reproduces multiplication by $\Delta_i$ inside the
integrand.

Applying this to the polynomial in curly brackets in \eqref{I1_restate}, we
obtain (\ref{I1_operator})
with $\mathcal O_{\rho,\gamma_E}$ given in \eqref{operator}. The factors
$\ln\rho_{ij}$ are treated as multiplicative coefficients because they depend
only on the external kinematics and are independent of the Mellin variables
$\Delta_i$.

\section{One-loop soft-kinematics calculation}
\label{app:soft_limit_C1}
Let
\begin{equation*}
S_1=s_1Q^2,\qquad S_2=s_2Q^2.
\end{equation*}
We have
\begin{align*}
\Xi_{12}^{(1)}
&=
-\frac{1}{2M^2Q^5}\Big[
M^2Q^4 - M^2Q^2(S_1+S_2) + 2M^2S_1S_2
\nonumber\\
&\hspace{60pt}+ 2Q^6 - 2Q^4(S_1+S_2) + Q^2S_1S_2
\Big]
\nonumber\\
&=
-\frac{1}{2Q}\big(1-s_1-s_2+2s_1s_2\big)
-\frac{Q}{2M^2}\big(2-2(s_1+s_2)+s_1s_2\big),
\\[6pt]
\Xi_{13}^{(1)}
&=
\frac{1}{2M^2Q^5}\Big[
M^2Q^4 - 4M^2Q^2S_1 - M^2Q^2S_2 + 2M^2S_1^2 + 2M^2S_1S_2
\nonumber\\
&\hspace{60pt}+ 2Q^6 - 2Q^4(S_1+S_2) + Q^2S_1^2 + Q^2S_1S_2
\Big]
\nonumber\\
&=
\frac{1}{2Q}\big(1-4s_1-s_2+2s_1^2+2s_1s_2\big)
+\frac{Q}{2M^2}\big(2-2(s_1+s_2)+s_1^2+s_1s_2\big),
\\[6pt]
\Xi_{23}^{(1)}
&=
\frac{1}{2M^2Q^5}\Big[
M^2Q^4 - M^2Q^2S_1 - 4M^2Q^2S_2 + 2M^2S_1S_2 + 2M^2S_2^2
\nonumber\\
&\hspace{60pt}+ 2Q^6 - 2Q^4(S_1+S_2) + Q^2S_1S_2 + Q^2S_2^2
\Big]
\nonumber\\
&=
\frac{1}{2Q}\big(1-s_1-4s_2+2s_1s_2+2s_2^2\big)
+\frac{Q}{2M^2}\big(2-2(s_1+s_2)+s_1s_2+s_2^2\big),
\end{align*}
and
\begin{align*}
Q\,\Xi_{12}^{(1)} &\longrightarrow -\frac12\big(1-s_1-s_2+2s_1s_2\big),\\
Q\,\Xi_{13}^{(1)} &\longrightarrow \frac12\big(1-4s_1-s_2+2s_1^2+2s_1s_2\big),\\
Q\,\Xi_{23}^{(1)} &\longrightarrow \frac12\big(1-s_1-4s_2+2s_1s_2+2s_2^2\big).
\end{align*}
Therefore,
\begin{align*}
Q\,{\sf C}_1
&=
\frac14\Big[
\ln\rho_{12}\,Q\Xi_{12}^{(1)}
+\ln\rho_{13}\,Q\Xi_{13}^{(1)}
+\ln\rho_{23}\,Q\Xi_{23}^{(1)}
\Big]
\nonumber\\
&\longrightarrow
\frac18\bigg[
-\ln\rho_{12}\,\big(1-s_1-s_2+2s_1s_2\big)
+\ln\rho_{13}\,\big(1-4s_1-s_2+2s_1^2+2s_1s_2\big)
\nonumber\\
&\hspace{5.5em}
+\ln\rho_{23}\,\big(1-s_1-4s_2+2s_1s_2+2s_2^2\big)
\bigg].
\end{align*}

\section*{Acknowledgments}
This research was funded by
National Science Centre, Poland (NCN) under grant no. 2023/49/B/ST2/01371.


\begin{thebibliography}{99}
\providecommand{\href}[2]{#2}
\providecommand{\arxivref}[2]{\href{http://arxiv.org/abs/#1}{#2}}
\providecommand{\doiref}[2]{\href{http://dx.doi.org/#1}{#2}}
\providecommand{\nbbstauthor}[1]{#1}
\providecommand{\nbbstjournal}[1]{\textsf{#1}}
\providecommand{\nbbsttitle}[1]{\textit{#1}}
\providecommand{\nbbsturl}[1]{\texttt{#1}}
\providecommand{\nbbsteprint}[1]{\texttt{#1}}
\providecommand{\nbbststyle}{\raggedright\small\parskip0pt}
\nbbststyle

\bibitem{Strominger2017}
\nbbstauthor{A.~Strominger},
\nbbsttitle{Lectures on the Infrared Structure of Gravity and Gauge Theory},
\nbbstjournal{arXiv:1703.05448}.

\bibitem{Pasterski:2021rjz} 
\nbbstauthor{S.~Pasterski, M.~Pate, A.~M.~Raclariu},
\nbbsttitle{Celestial Holography},
\nbbstjournal{arXiv:2111.11392}.

\bibitem{Raclariu:2021zjz}
\nbbstauthor{A.~M.~Raclariu},
\nbbsttitle{Lectures on Celestial Holography},
\nbbstjournal{arXiv:2107.02075}.

\bibitem{Donnay:2023celestial}
\nbbstauthor{L.~Donnay},
\nbbsttitle{Celestial Holography: An Asymptotic Symmetry Perspective},
\nbbstjournal{Phys.~Rept.~1073 (2024) 1-41, arXiv:2310.12922}.

\bibitem{Pasterski:2021dqe}
\nbbstauthor{S.~Pasterski},
\nbbsttitle{Lectures on celestial amplitudes},
\nbbstjournal{Eur.~Phys.~J.~C81 (2021) 12, 1062, arXiv:2108.04801}.

\bibitem{Pasterski:2017kqt}
\nbbstauthor{S.~Pasterski, S.~H.~Shao, A.~Strominger},
\nbbsttitle{Flat Space Amplitudes and Conformal Symmetry of the Celestial Sphere},
\nbbstjournal{Phys.~Rev.~D96 (2017) 065026, arXiv:1701.00049}.

\bibitem{Schreiber:2017jsr}
\nbbstauthor{A.~Schreiber, A.~Volovich, M.~Zlotnikov},
\nbbsttitle{Tree-level gluon amplitudes on the celestial sphere},
\nbbstjournal{Phys.~Lett.~B781 (2018) 349-357, arXiv:1711.08435}.

\bibitem{Kalyanapuram:2020aya}
\nbbstauthor{N.~Kalyanapuram},
\nbbsttitle{Gauge and Gravity Amplitudes on the Celestial Sphere},
\nbbstjournal{Phys.~Rev.~D103 (2021) 085015, arXiv:2012.04579}.

\bibitem{Banerjee:2020kaa}
\nbbstauthor{S.~Banerjee, S.~Ghosh, P.~Paul},
\nbbsttitle{MHV Graviton Scattering Amplitudes and Current Algebra on the Celestial Sphere},
\nbbstjournal{JHEP 02 (2021) 176, arXiv:2008.04330}.

\bibitem{Melton:2024akx} 
\nbbstauthor{W.~Melton, A.~Sharma, A.~Strominger, T.~Wang},
\nbbsttitle{Celestial Dual for Maximal Helicity Violating Amplitudes},
\nbbstjournal{Phys.~Rev.~Lett.~133 (2024) 9, 091603, arXiv:2403.18896}.

\bibitem{Giribet:2024vnk}  
\nbbstauthor{G.~Giribet}, 
\nbbsttitle{Remarks on celestial amplitudes and Liouville theory},
\nbbstjournal{Int.~J.~Mod.~Phys.~D34 (2025) 01, arXiv:2403.03374}.

\bibitem{Donnay:2025yoy}
\nbbstauthor{L.~Donnay, G.~Giribet, B.~Valsesia},
\nbbsttitle{MHV leaf amplitudes from parafermions},
\nbbstjournal{JHEP 06 (2025) 234, arXiv:2501.19332}.

\bibitem{Stieberger2023a}
\nbbstauthor{S.~Stieberger, T.R.~Taylor, B.~Zhu}, 
\nbbsttitle{Celestial Liouville theory for Yang--Mills amplitudes}, 
\nbbstjournal{Phys.~Lett.~B836 (2023) 137588}.

\bibitem{Stieberger2023b}
\nbbstauthor{S.~Stieberger, T.R.~Taylor, B.~Zhu}, 
\nbbsttitle{Yang--Mills as a Liouville theory}, 
\nbbstjournal{Phys.~Lett.~B846 (2023) 138229}.

\bibitem{BPZ} 
\nbbstauthor{A.A.~Belavin, A.M.~Polyakov, A.B.~Zamolodchikov}, 
\nbbsttitle{Infinite conformal symmetry in 2D quantum field theories}, 
\nbbstjournal{Nucl.~Phys.~B241 (1984) 333}.

\bibitem{S}
\nbbstauthor{N.~Seiberg}, 
\nbbsttitle{Notes on Quantum Liouville Theory and Quantum Gravity}, 
\nbbstjournal{Prog.~Theor.~Phys.~Suppl.~102 (1990) 319-349}.

\bibitem{DO} 
\nbbstauthor{H.~Dorn, H.J.~Otto}, 
\nbbsttitle{Two and three point functions in Liouville theory}, 
\nbbstjournal{Nucl.~Phys.~B429 (1994) 375-388, hep-th/9403141}.

\bibitem{T1} 
\nbbstauthor{J.~Teschner}, 
\nbbsttitle{On the Liouville three point function}, 
\nbbstjournal{Phys.~Lett.~B363 (1995) 65-70, hep-th/9507109}.

\bibitem{ZZ5} 
\nbbstauthor{A.B.~Zamolodchikov, A.B.~Zamolodchikov}, 
\nbbsttitle{Structure constants and conformal bootstrap in Liouville field theory}, 
\nbbstjournal{Nucl.~Phys.~B477 (1996) 577, hep-th/9506136}.

\bibitem{T3}
\nbbstauthor{J.~Teschner},
\nbbsttitle{Liouville theory revisited}, 
\nbbstjournal{Class.~Quant.~Grav.~18 (2001) R153, hep-th/0104158}.

\bibitem{T4}
\nbbstauthor{J.~Teschner},
\nbbsttitle{A lecture on the Liouville vertex operators}, 
\nbbstjournal{Int.~J.~Mod.~Phys.~A19 (2004) 436-458, hep-th/0303150}.

\bibitem{Nakayama:2004vk}
\nbbstauthor{Y.~Nakayama},
\nbbsttitle{Liouville field theory: A Decade after the revolution}
\nbbstjournal{Int.~J.~Mod.~Phys.~A19 (2004) 2771-2930, hep-th/0402009}.

\bibitem{HMW}
\nbbstauthor{D.~Harlow, J.~Maltz, E.~Witten},
\nbbsttitle{Analytic Continuation of Liouville Theory},
\nbbstjournal{JHEP 1112 (2011) 071, arXiv:1108.4417}.

\bibitem{SR14}
\nbbstauthor{S.~Ribault}, 
\nbbsttitle{Conformal field theory on the plane}, 
\nbbstjournal{arXiv:1406.4290}.

\bibitem{Thorn}
\nbbstauthor{Ch.B.~Thorn},
\nbbsttitle{Liouville Perturbation Theory},	
\nbbstjournal{Phys.~Rev.~D66 (2002) 027702, arXiv:hep-th/0204142}.	

\bibitem{Ferrari2025}
\nbbstauthor{F.~Ferrari, M.~R.~Pi\c{a}tek, A.~R.~Pietrykowski},
\nbbsttitle{Small-b expansion of the DOZZ formula for light operators},
\nbbstjournal{arXiv:2509.21182v2}.

%\bibitem{HJ04} 
%\nbbstauthor{L.~Hadasz, Z.~Jask\'{o}lski}, 
%\nbbsttitle{Classical Liouville action on the sphere with three hyperbolic singularities}, 
%\nbbstjournal{Nucl.~Phys.~B694 (2004) 493, hep-th/0309267}. 
%
%\bibitem{HJP05}
%\nbbstauthor{L.~Hadasz, Z.~Jask\'{o}lski, M.~Pi\c{a}tek}, 
%\nbbsttitle{Classical geometry from the quantum Liouville theory},
%\nbbstjournal{Nucl.~Phys.~B724 (2005) 529, hep-th/0504204}.

%\bibitem{Pasterski:2021rjz} 
%\nbbstauthor{S.~Pasterski, M.~Pate, A.~M.~Raclariu},
%\nbbsttitle{Celestial Holography},
%\nbbstjournal{arXiv:2111.11392}.
%
%\bibitem{Raclariu:2021zjz}
%\nbbstauthor{A.~M.~Raclariu},
%\nbbsttitle{Lectures on Celestial Holography},
%\nbbstjournal{arXiv:2107.02075}.
%
%\bibitem{Donnay:2023celestial}
%\nbbstauthor{L.~Donnay},
%\nbbsttitle{Celestial Holography: An Asymptotic Symmetry Perspective},
%\nbbstjournal{Phys.~Rept.~1073 (2024) 1-41, arXiv:2310.12922}.
%
%\bibitem{Pasterski:2021dqe}
%\nbbstauthor{S.~Pasterski},
%\nbbsttitle{Lectures on celestial amplitudes},
%\nbbstjournal{Eur.~Phys.~J.~C81 (2021) 12, 1062, arXiv:2108.04801}.
%
%\bibitem{Pasterski:2017kqt}
%\nbbstauthor{S.~Pasterski, S.~H.~Shao, A.~Strominger},
%\nbbsttitle{Flat Space Amplitudes and Conformal Symmetry of the Celestial Sphere},
%\nbbstjournal{Phys.~Rev.~D96 (2017) 065026, arXiv:1701.00049}.
%
%\bibitem{Schreiber:2017jsr}
%\nbbstauthor{A.~Schreiber, A.~Volovich, M.~Zlotnikov},
%\nbbsttitle{Tree-level gluon amplitudes on the celestial sphere},
%\nbbstjournal{Phys.~Lett.~B781 (2018) 349-357, arXiv:1711.08435}.
%
%\bibitem{Kalyanapuram:2020aya}
%\nbbstauthor{N.~Kalyanapuram},
%\nbbsttitle{Gauge and Gravity Amplitudes on the Celestial Sphere},
%\nbbstjournal{Phys.~Rev.~D103 (2021) 085015, arXiv:2012.04579}.
%
%\bibitem{Banerjee:2020kaa}
%\nbbstauthor{S.~Banerjee, S.~Ghosh, P.~Paul},
%\nbbsttitle{MHV Graviton Scattering Amplitudes and Current Algebra on the Celestial Sphere},
%\nbbstjournal{JHEP 02 (2021) 176, arXiv:2008.04330}.
%
%\bibitem{Stieberger2023a}
%\nbbstauthor{S.~Stieberger, T.R.~Taylor, B.~Zhu}, 
%\nbbsttitle{Celestial Liouville theory for Yang--Mills amplitudes}, 
%\nbbstjournal{Phys.~Lett.~B836 (2023) 137588}.
%
%\bibitem{Melton:2024akx}
%\nbbstauthor{W.~Melton, A.~Sharma, A.~Strominger, T.~Wang},
%\nbbsttitle{Celestial Dual for Maximal Helicity Violating Amplitudes},
%\nbbstjournal{Phys.~Rev.~Lett.~133 (2024) 9, 091603, arXiv:2403.18896}.

%\bibitem{Giribet:2024vnk}  
%\nbbstauthor{G.~Giribet}, 
%\nbbsttitle{Remarks on celestial amplitudes and Liouville theory},
%\nbbstjournal{Int.~J.~Mod.~Phys.~D34 (2025) 01, arXiv:2403.03374}.
%
%\bibitem{Donnay:2025yoy}
%\nbbstauthor{L.~Donnay, G.~Giribet, B.~Valsesia},
%\nbbsttitle{MHV leaf amplitudes from parafermions},
%\nbbstjournal{JHEP 06 (2025) 234, arXiv:2501.19332}.

%\bibitem{Perlmutter:2015}
%\nbbstauthor{E.~Perlmutter},
%\nbbsttitle{Virasoro conformal blocks in closed form},
%\nbbstjournal{JHEP 02 (2015) 023, arXiv:1502.07742}.
%
%\bibitem{Bombini:2018jrg}
%\nbbstauthor{A.~Bombini, S.~Giusto, R.~Russo},  
%\nbbsttitle{A note on the Virasoro blocks at order $1/c$}, 
%\nbbstjournal{Eur.~Phys.~J.~C79, 3 (2019), arXiv:1807.07886}.
\end{thebibliography}
\end{document}